\documentclass[12pt]{article}
\newcommand{\blind}{1}
\usepackage{amsthm}
\usepackage{geometry}
\usepackage[nodisplayskipstretch]{setspace}
\usepackage{csquotes}

\PassOptionsToPackage{
        natbib=true,
        style=authoryear-comp,
        hyperref=true,
        backend=bibtex,
        maxbibnames=1,
        giveninits=true,
        uniquename=init,
        maxcitenames=1,
        parentracker=true,
        url=false,
        doi=false,
        isbn=false,
        eprint=false,
        backref=false,
            }   {biblatex}
\usepackage{biblatex}
\DeclareNameAlias{sortname}{family-given}

\renewbibmacro{in:}{%
  \ifentrytype{article}{}{%
  \printtext{\bibstring{in}\intitlepunct}}}
  
\AtEveryBibitem{%
  \clearfield{month}{}%
  \clearlist{language}{}%
  }

\let\cite\textcite

\renewbibmacro*{volume+number+eid}{%
  \setunit*{\addcomma\space}
  \printfield{volume}%
  \printfield{number}%
  \printfield{eid}}
  \DeclareFieldFormat[article]{number}{(#1)}

\makeatletter
\renewbibmacro*{cite}{
  \iffieldundef{shorthand}
    {\ifthenelse{\ifnameundef{labelname}\OR\iffieldundef{labelyear}}
       {\printtext[bibhyperref]{
          \DeclareFieldAlias{bibhyperref}{default}
          \usebibmacro{cite:label}%
          \setunit{\addspace}%
          \usebibmacro{cite:labeldate+extradate}}%
          \usebibmacro{cite:reinit}}
       {\iffieldequals{namehash}{\cbx@lasthash}
          {\ifthenelse{\iffieldequals{labelyear}{\cbx@lastyear}\AND
                       \(\value{multicitecount}=0\OR\iffieldundef{postnote}\)}
             {\setunit{\addcomma}
              \usebibmacro{cite:extrayear}}
             {\setunit{\compcitedelim}%
              \usebibmacro{cite:labeldate+extradate}%
              \savefield{labelyear}{\cbx@lastyear}}}
          {\printtext[bibhyperref]{
             \DeclareFieldAlias{bibhyperref}{default}
             \printnames{labelname}%
             \setunit{\nameyeardelim}%
             \usebibmacro{cite:labeldate+extradate}}%
             \savefield{namehash}{\cbx@lasthash}%
             \savefield{labelyear}{\cbx@lastyear}}}}
    {\usebibmacro{cite:shorthand}%
     \usebibmacro{cite:reinit}}%
  \setunit{\multicitedelim}}

\renewbibmacro*{textcite}{
  \iffieldequals{namehash}{\cbx@lasthash}
    {\iffieldundef{shorthand}
       {\ifthenelse{\iffieldequals{labelyear}{\cbx@lastyear}\AND
                    \(\value{multicitecount}=0\OR\iffieldundef{postnote}\)}
          {\setunit{\addcomma}%
           \usebibmacro{cite:extrayear}}
          {\setunit{\compcitedelim}%
           \usebibmacro{cite:labeldate+extradate}%
           \savefield{labelyear}{\cbx@lastyear}}}
       {\setunit{\compcitedelim}%
        \usebibmacro{cite:shorthand}%
        \global\undef\cbx@lastyear}}
    {\ifnameundef{labelname}
       {\printtext[bibhyperref]{
          \DeclareFieldAlias{bibhyperref}{default}
          \iffieldundef{shorthand}
            {\usebibmacro{cite:label}%
             \setunit{%
               \global\booltrue{cbx:parens}%
               \addspace\bibopenparen}%
             \ifnumequal{\value{citecount}}{1}
               {\usebibmacro{prenote}}
               {}%
             \usebibmacro{cite:labeldate+extradate}}
            {\usebibmacro{cite:shorthand}}%
          \ifthenelse{\iffieldundef{postnote}\AND
                      \(\value{multicitetotal}=0\AND\value{citetotal}=1\)}
            {\bibcloseparen
             \global\boolfalse{cbx:parens}}
            {}}}
       {\printtext[bibhyperref]{
          \DeclareFieldAlias{bibhyperref}{default}
          \printnames{labelname}%
          \setunit{%
            \global\booltrue{cbx:parens}%
            \addspace\bibopenparen}%
          \ifnumequal{\value{citecount}}{1}
            {\usebibmacro{prenote}}
            {}%
          \iffieldundef{shorthand}
            {\iffieldundef{labelyear}
               {\usebibmacro{cite:label}}
               {\usebibmacro{cite:labeldate+extradate}}%
             \savefield{labelyear}{\cbx@lastyear}}
            {\usebibmacro{cite:shorthand}%
             \global\undef\cbx@lastyear}%
          \ifthenelse{\iffieldundef{postnote}\AND
                      \(\value{multicitetotal}=0\AND\value{citetotal}=1\)}
            {\bibcloseparen
             \global\boolfalse{cbx:parens}}
            {}}%
          \savefield{namehash}{\cbx@lasthash}}}%
  \setunit{%
    \ifbool{cbx:parens}
      {\bibcloseparen\global\boolfalse{cbx:parens}}
      {}%
    \multicitedelim}}

\makeatother

\bibliography{ref.bib, ref_2.bib, ref_HSL_1.bib, ref_HSL_2.bib}
\usepackage{amsmath}
\usepackage{mathtools}
\usepackage[table, dvipsnames]{xcolor}
\definecolor{pathway1}{HTML}{E6194B}
\definecolor{pathway2}{HTML}{3CB44B}
\definecolor{pathway3}{HTML}{FFE119}
\definecolor{pathway4}{HTML}{4363D8}
\definecolor{pathway5}{HTML}{F58231}
\definecolor{pathway6}{HTML}{911EB4}
\definecolor{pathway7}{HTML}{46F0F0}
\definecolor{pathway8}{HTML}{F032E6}
\definecolor{pathway9}{HTML}{BCF60C}
\definecolor{pathway10}{HTML}{FABEBE}
\definecolor{pathway11}{HTML}{008080}
\definecolor{pathway12}{HTML}{E6BEFF}
\usepackage[normalem]{ulem}
\usepackage{graphicx}
\definecolor{darkblue}{rgb}{0.0,0.0,0.6}
\usepackage[colorlinks,citecolor=darkblue,urlcolor=darkblue, linkcolor=darkblue]{hyperref} 
\hypersetup{breaklinks=true}
\usepackage{chngcntr}
\newcounter{mysfig}
\counterwithin{mysfig}{figure}
\renewcommand\themysfig{(\alph{mysfig})}
\makeatletter
\newcommand\Scaption[1]{%
\refstepcounter{mysfig}%
  \sbox\@tempboxa{\small\themysfig~#1}%
  \ifdim \wd\@tempboxa >\hsize
    \small\themysfig~#1\par
  \else
    \global \@minipagefalse
    \hb@xt@\hsize{\hfil\box\@tempboxa\hfil}%
  \fi
  }
\makeatother
\usepackage{etoolbox}
\makeatletter
\patchcmd{\@makecaption}
  {\parbox}
  {\advance\@tempdima-\fontdimen2} 
  {}{}
\makeatother 
\newtheorem{theorem}{Theorem}[section]

\newtheorem{lemma}[theorem]{Lemma}

\input{amssym}
\usepackage{bm}

\newcommand{\Evec}{\ensuremath{\bm{E}}}

\newcommand{\Vvec}{\ensuremath{\bm{V}}}

\newcommand{\Xvec}{\ensuremath{\bm{X}}}
\newcommand{\Yvec}{\ensuremath{\bm{Y}}}

\newcommand{\alphavec}{\ensuremath{{\bm{\alpha}}}}
\newcommand{\betavec}{\ensuremath{{\bm{\beta}}}}

\newcommand{\thetavec}{\ensuremath{{\bm{\theta}}}}

\newcommand{\xivec}{\ensuremath{{\bm{\xi}}}}

\newcommand{\Dmat}{\ensuremath{\bm{D}}}

\newcommand{\Imat}{\ensuremath{\bm{I}}}

\newcommand{\Qmat}{\ensuremath{\bm{Q}}}

\newcommand{\Umat}{\ensuremath{\bm{U}}}

\newcommand{\Xmat}{\ensuremath{\bm{X}}}
\newcommand{\Ymat}{\ensuremath{\bm{Y}}}

\newcommand{\be}{\begin{equation}}
\newcommand{\ee}{\end{equation}}
\newcommand{\beqa}{\begin{eqnarray*}}
\newcommand{\eeqa}{\end{eqnarray*}}
\newcommand{\beqn}{\begin{eqnarray}}
\newcommand{\eeqn}{\end{eqnarray}}
\newcommand{\ba}{\begin{array}}
\newcommand{\ea}{\end{array}}
\newcommand{\bc}{\begin{center}}
\newcommand{\ec}{\end{center}}
\newcommand{\btab}{\begin{tabular}}
\newcommand{\etab}{\end{tabular}}

\newcommand{\mb}{\makebox}

\newcommand{\st}{\stackrel}

\newcommand{\Ind}{1\!\mathrm{l}}

\newcommand{\ind}{\, \raise-2pt\hbox{$\st{\mb{\scriptsize ind}}{\sim}$}\, }
\newcommand{\iid}{\, \raise-2pt\hbox{$\st{\mb{\scriptsize iid}}{\sim}$}\,}

\newcommand{\bX}{\bm{X}}

\newcommand{\bY}{\bm{Y}}

\newcommand{\bgamma}{\bm{\gamma}}

\newcommand{\bSigma}{\bm{\Sigma}}

\mathchardef\given="626A
\mathcode`:="603A

\long\def\beginskip#1\endskip{}
\def\endskip{}

\newtheorem{condition}[theorem]{Condition}

\numberwithin{theorem}{section}

\usepackage{bold-extra}

\newcommand{\edge}[2]{\texttt{#1}$\leftarrow$\texttt{#2}}
\allowdisplaybreaks
\begin{document}

\def\spacingset#1{\renewcommand{\baselinestretch}%
{#1}\small\normalsize} \spacingset{1}


\if1\blind
{
  \title{\bf Bayesian Covariate-Dependent Quantile Directed Acyclic Graphical Models for Individualized Inference}
  \author{Ksheera Sagar \\ 
    Department of Statistics, Purdue University \\ 
    and\\
    Yang Ni\\
    Department of Statistics, Texas A\&M University\\
    and\\
    Veerabhadran Baladandayuthapani\\
    Department of Biostatistics, University of Michigan\\
    and\\
    Anindya Bhadra\thanks{Address for correspondence: 150 N. University St., West Lafayette, IN 47906. Email: bhadra@purdue.edu}\\
    Department of Statistics, Purdue University}
    \date{}
  \maketitle
} \fi

\if0\blind
{
  \bigskip
  \bigskip
  \bigskip
  \begin{center}
    {\LARGE\bf Bayesian Covariate-Dependent Quantile Directed Acyclic Graphical Models for Individualized Inference}
\end{center}
  \medskip
} \fi

\bigskip
\begin{abstract}
We propose an approach termed ``qDAGx'' for Bayesian covariate-dependent quantile directed acyclic graphs (DAGs) where these DAGs are individualized, in the sense that they depend on individual-specific covariates. The individualized DAG structure of the proposed approach  can be uniquely identified at any given quantile, based on purely observational data without strong assumptions such as a known topological ordering. To scale the proposed method to a large number of variables and covariates, we propose for the model parameters a novel parameter expanded horseshoe prior that affords a number of attractive theoretical and computational benefits to our approach. By modeling the conditional quantiles, qDAGx overcomes the common limitations of mean regression for DAGs, which can be sensitive to the choice of likelihood, e.g., an assumption of multivariate normality, as well as to the choice of priors. We demonstrate the performance of qDAGx through extensive numerical simulations and via an application in precision medicine, which infers patient-specific protein--protein interaction networks in lung cancer.
\end{abstract}

\noindent%
{\it Keywords:}  DAG identifiability, Global-local shrinkage priors, Precision medicine, Quantile regression, Varying sparsity model.
\vfill

\newpage
\doublespacing
\section{Introduction}
Graphs are one of the most common tools for studying associations between variables in multivariate data. A graph can be denoted as an ordered pair $\mathcal{G}=\left(\Vvec,\Evec\right)$, where $\Vvec$ denotes a set of nodes and $\Evec$ denotes the set of edges or associations between these nodes. If the data comprise of continuous variables, the literature on embedding a probabilistic model to study the graph structure is rich~\citep{lauritzen1996, jordan2004graphical, koller2009probabilistic}. Under a topological ordering of the nodes, the most popular among these probabilistic approaches is the Gaussian directed acyclic graphs (DAGs). Despite this popularity, the limitations of Gaussian DAGs include a specific assumption on the likelihood and the consequent lack of robustness to model misspecification. We address this by proposing a technique that circumvents the Gaussian likelihood assumption and can model association between variables at any given quantile level, $\tau\in(0,\,1)$. Further, typical approaches for DAG inference are global, in the sense that all samples are aggregated to produce a single DAG estimate at the population level, which could be a modeling limitation when \emph{individualized inference} is a central goal, for example, in  \emph{precision medicine}. Our approach remedies this situation by modeling the conditional quantiles as a function of individual-specific covariates, which then allows the inferred quantile DAGs to vary across samples. 

Using a topological sort, one can permute the nodes of a DAG, which then renders the corresponding adjacency matrix of $\mathcal{G}$ strictly upper triangular. This permutation, which is not necessarily unique, is generally known as an \emph{ordering} of the nodes. Given $p$ variables or nodes, $\Ymat =\left(\Yvec_1,\,\ldots,\Yvec_p\right)$ in a DAG, the ordering of the nodes implies, for every directed edge $\Yvec_{h}\leftarrow \Yvec_{j}$ $\forall\, h,\, j\in\{1,\ldots,p\} \text{ and }h\neq j$, $\Yvec_{h}$ appears \emph{before} $\Yvec_{j}$ in the order. In Gaussian models, \citet{wu2003nonparametric} and \citet{ huang2006covariance} have worked with longitudinal data and time-ordered random vectors respectively, in which the ordering of the nodes is natural. However, there are several other applications where a natural ordering may not exist, for example on a lattice on an isotropic random field. The usual remedy is to  impose a \emph{fixed} ordering of nodes in a multivariate Gaussian model to infer the structure of the underlying DAG~\citep{drton2008sinful, altomare2013objective, ni2019bayesian, makam2021symmetries}. However, an inference procedure relying on an imposed ordering can be sensitive to its mis-specification.

A generic way to infer Gaussian DAGs is to consider the joint likelihood $\pi(\Yvec) = \prod_{h=1}^{p}\pi(\Yvec_h \mid pa(h))$ and to estimate the coefficients $\beta_{hj}$ in the corresponding node conditional regression: $Y_{ih}=\sum_{j\in pa(h)}\beta_{hj}Y_{ij}+\varepsilon_{ih}$, where $\Yvec_h = (Y_{1h},\ldots,Y_{nh})^T$, $i\in\{1,\ldots,n\}$, $\varepsilon_{ih}\sim\mathcal{N}(0,\sigma_h^2)$ and $pa(h)$ denotes the parent set of $\Yvec_h$ i.e., the set of nodes $\Yvec_j$s for which there exists an edge $\Yvec_h\leftarrow\Yvec_j$. With these rudimentary notations in place, we are now
poised to summarize some attributes of a Gaussian DAG model more precisely; and point out some cases where these could be limitations. 
\begin{enumerate}
\item[(a)] An assumption of Gaussian likelihood, like any other likelihood based modeling framework, is susceptible to model mis-specification; in the sense that a mis-specified likelihood could lead to wrong inference or poor performance.

\item[(b)] The ordering  of nodes is typically assumed known, i.e., $pa(h)$ is constrained to be a subset of the nodes that appear later than $\Yvec_h$ in the order, which could be arbitrary in some applications. In the absence of a given ordering, although techniques such as the PC algorithm~\citep{spirtes2000causation} could be used to extract a (partial) ordering, the inferred DAG is not unique, and inference could become conditional on the specific extracted ordering.

\item[(c)] Yet another feature of a Gaussian likelihood assumption is that $\beta_{hj}$ cannot be estimated only at a given quantile of interest $\tau\in(0,1)$ without modeling the entire distribution. This could be a limitation when the primary objective is  modeling certain quantiles, typical for example in diseases such as cancer, which are often characterized by genomic, transcriptomic and proteomic changes relative to the control group~\citep{zhang1997gene} at higher or lower quantiles, while exhibiting similar behavior towards the center of the distribution. In such situations,  a quantile-based approach could be better suited to delineate certain parts of the distribution, without modeling the entire distribution.


\item[(d)]  Finally, for a Gaussian DAG, the coefficient ${\beta}_{hj}\neq 0$ signifies an edge $\Yvec_h\leftarrow\Yvec_j$ for all observations $i\in\{1,\ldots,n\}$, which does not allow for individualized inference on these coefficients. This aspect is unappealing for many modern applications, including \emph{precision medicine}, where the focus may be on inferring a protein--protein interaction network whose structure and strength depend on individual genotypes or other subject-specific covariates.
\end{enumerate}

There are some existing approaches that partially address some of these limitations. For example, the linear non-Gaussian acyclic model \citep{shimizu2006linear} addresses limitation (b) by taking advantage of the identifiability theory of independent component analysis. While this method considers non-Gaussian errors, it still proceeds by modeling the conditional expectations, which is only possible for certain classes of likelihood. Hence, this approach only partially addresses limitation (a), which is also true for many other identifiable DAG models \citep[e.g.,][]{hoyer2008nonlinear}. 
The varying coefficient model \citep{hastie1993varying} provides one approach for modeling individualized coefficients. Inspired by this approach, \citet{ni2019bayesian} have addressed limitation (d) in undirected Gaussian graphical models by modeling $\beta_{hj}$ as $\beta_{hj}(\Xvec_i)$, where $\Xvec_i$ corresponds to observation-specific external covariates. But limitations (a)--(c) listed above still 
remain.
To our knowledge, there do not exist DAG methods that simultaneously address all the aforementioned limitations.

Hence, to simultaneously overcome limitations (a)--(d), we propose a new model termed ``qDAGx" that learns quantile directed acyclic graphs (quantile-DAGs), $\mathcal{Q}\mathcal{G}_i^{(\tau)}=(\Vvec_i,\,\Evec_i^{(\tau)})$, at any given quantile $\tau\in(0,1)$, where the conditional quantile functions depend on individual-specific covariates ($\Xvec_i$). Our model is robust to likelihood mis-specification, does not assume a known ordering of the nodes, and infers individualized, quantile-specific DAGs. Figure~\ref{intitution_pic} provides some intuition for the modeling principle followed in qDAGx (with one covariate as an example), for which the DAG structure and edge strength vary smoothly as a function of an individual-specific covariate at several quantiles.

\begin{figure}[!t]
    \centering
    \includegraphics[width = \textwidth]{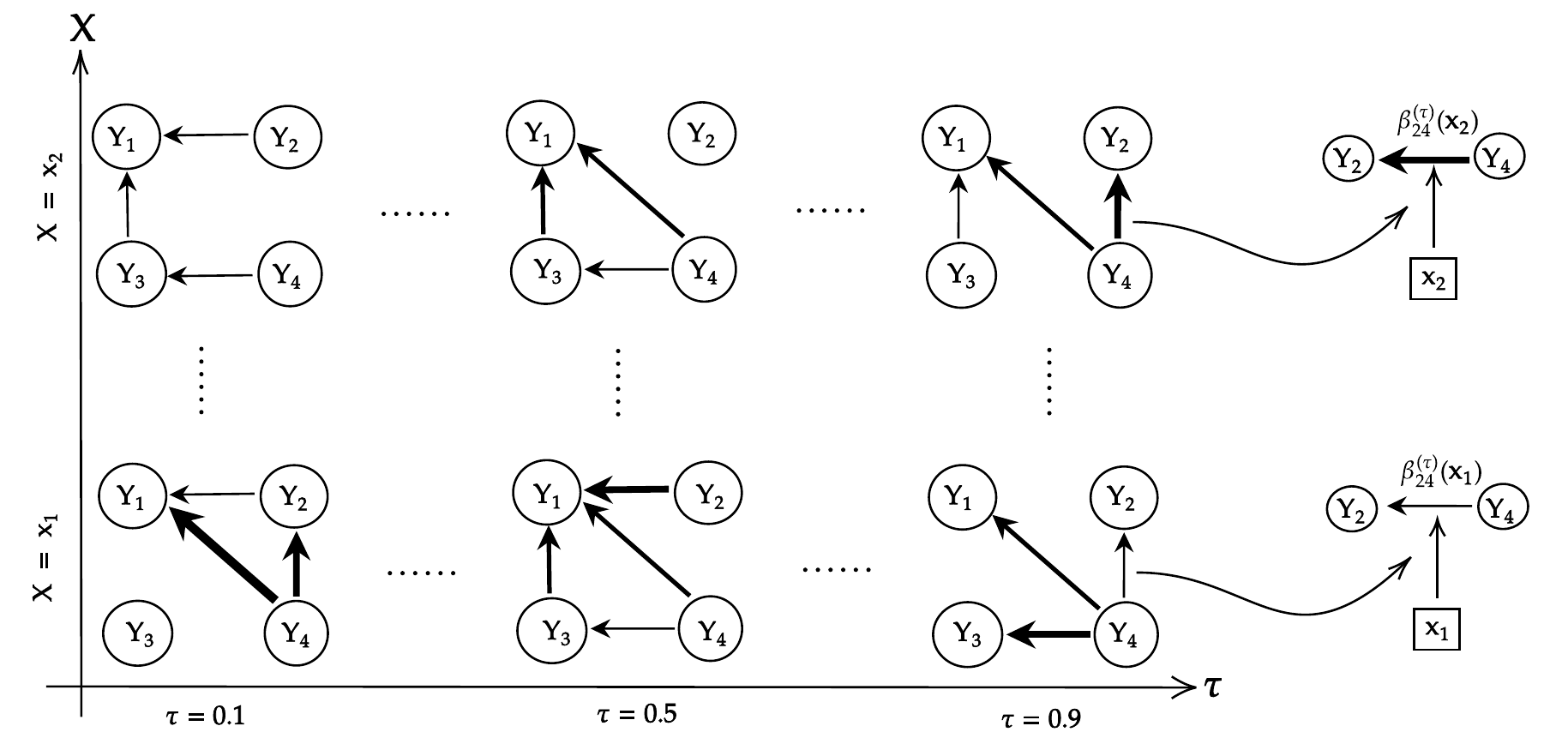}
    \caption{\label{intitution_pic} Schematic representation of qDAGx.  The directed acyclic graphs $\mathcal{Q}\mathcal{G}_i^{(\tau)}$, for $n=2$ observations, on $p=4$ vertices $\Yvec=(\Yvec_1,\,\Yvec_2,\Yvec_3,\,\Yvec_4)$ are presented, for univariate $\Xvec_i,\; i=1,\ldots,n;$ and for given quantile levels $\tau=0.1,\,0.5\text{ and }0.9$. Thickness of an edge $\Yvec_h\leftarrow\Yvec_j;\,h,\,j\in\{1,\ldots,p\}$ and $h\neq j$ represents the strength of  association between $\Yvec_h$ and $\Yvec_j$ and the direction denotes the influence of a parent ($\Yvec_j$) on its child ($\Yvec_h$). The figure further shows the individualized nature of qDAGx, i.e., each directed edge between the nodes is influenced by external covariates $\Xvec_i$.  If the coefficient $\beta_{24}^{(\tau)}(\Xvec_i)\neq 0$, then there exists an edge from $\Yvec_4$ to $\Yvec_2$ at quantile $\tau$ for individual $i$.
    }
\end{figure}

From an applied perspective,  learning DAGs provides crucial tools for understanding the topology of gene regulatory networks and protein--protein interaction networks \citep{segal2003module,friedman2004inferring, mallick2009bayesian}. However, their applications in the context of {personalized inference} that is robust to a {mis-specified likelihood} is limited, which we seek to address in the current work through an application of our methodology in precision medicine by inferring patient-specific protein--protein interaction networks in lung cancer. Moreover, it is worth mentioning at this juncture that apart from precision medicine, our proposed approach is applicable far more broadly to applications where similar individualized inference may be of interest. For example, in infectious disease epidemiology, an important problem is the study of contact networks, which are naturally individual-specific. Further, the edge strengths in such networks could be  modeled as a function of individual-specific covariates such as age or other individual risk factors  \citep{keeling2005networks}.

\subsection{Related works on quantile graphical models and foundations of qDAGx}
\label{foundations_qDAGx}
The inference for qDAGx proceeds via modeling the association between a node $\Yvec_h$ and its parents through a conditional quantile function. Introduced in the seminal work of \citet{koenker1978regression}, quantile regression (QR) has found applications in diverse domains, particularly in economics, management and quantitative finance \citep{chamberlain1994quantile, yu2003quantile, li2015moving, koenker2017quantile}. We refer the readers to  \cite{briollais2014application} for applications of QR in recent genetic and -omic studies. In the Bayesian setting, \citet{tsionas2003bayesian} demonstrated the equivalence of QR with location-scale mixture of normals leading to an asymmetric Laplace distribution. This representation enables a data augmented Gibbs sampler, facilitating the study  of penalized QR in various settings \citep[see][and references therein]{li2010bayesian}. \citet{sriram2013posterior} have studied the posterior consistency of Bayesian QR, based on a mis-specified asymmetric Laplace density. A recent work inspired by varying coefficient models in quantile regression is by \citet{QUANTICO}, but they do not consider graphical models.  Inferring undirected quantile graphical models with penalized QR has been studied by \citet{guha2020quantile}, including its consistency properties. Another work in the area of penalized quantile graphical models is by \citet{NIPS2016_537de305}, who estimate undirected quantile graphs at multiple quantile levels, using pseudolikelihood-based approach. But, their approach can neither infer DAGs nor provide individualized quantile graphs like the proposed qDAGx. Similarly, although \cite{guha2020quantile}  deal with quantile graphs, their model has the same drawbacks as \citet{NIPS2016_537de305}, when compared with qDAGx. 
\subsection{Summary of our contributions and organization of the article}
Our key contributions in this paper can be summarized as follows.
\begin{enumerate}
    \item[(a)]\textbf{Methodological:} We propose a technique for learning individual-specific DAGs at any quantile level $\tau\in(0,\,1)$, with no assumptions on the data likelihood or on the ordering of nodes. These features enable us to capture associations between variables for every observation at any quantile level $\tau$ (see Figure~\ref{intitution_pic}). Individualized inference in our modeling framework is possible by estimating functional forms for the edge associations (e.g., $\beta_{24}^{(\tau)}(\bX_i)$ in Figure~\ref{intitution_pic}), as specified explicitly later in Equations~\eqref{beta_hj} and~\eqref{theta_formulation}. Estimation proceeds via ``borrowing strength'' across all observations to infer population level parameters. Individual-specific inference is then possible by using covariate $\bX_i$ for the $i$th individual in the estimated functional forms, which are parameterized by these population-level parameters. 
    
    \item[(b)]\textbf{Theoretical:} We establish structural identifiability of the quantile-DAGs inferred by qDAGx by showing under mild conditions on the model parameters that the inferred quantile DAG is unique. We further prove a \emph{non-local} property of our prior, which aids in sparse quantile DAG discovery, and also prove posterior consistency of the fitted density at any node $\Yvec_h,\, h\in\{1,\ldots,p\}$.  
    
    \item[(c)]\textbf{Applied:} As an application of qDAGx to \emph{precision medicine}, we infer \emph{individual-specific} protein--protein interaction networks in patients with lung adenocarcinoma and lung squamous cell carcinoma. Although there exist population-level network inference, individualized inference similar to ours has been unexplored in the literature so far. We model the protein--protein association in each patient at a quantile level $\tau$, as a function of external covariates mRNA and methylation, which are known to play an important role in the disease. 
\end{enumerate}

A brief outline of the remainder of the paper is as follows. In Section~\ref{qDAGx_master_intro_section}, we introduce the proposed qDAGx model. Theoretical properties, including model identifiability, a  \emph{non-local} property of the prior, and posterior consistency are discussed in Section~\ref{theorectical_properties_qDAGx}. We provide numerical results in Section~\ref{simulation_results}, demonstrating the performance of qDAGx in learning individualized quantile-DAGs. Section~\ref{real_data_application_qDAGx} discusses the aforementioned application of qDAGx in precision medicine. We conclude in Section~\ref{conclusion_qDAGx} discussing some potential future directions.  

\section{qDAGx: Individualized Quantile Graphical Models with DAG Learning}
\label{qDAGx_master_intro_section}
Let the $p$ response variables be denoted as $\Yvec_1, \ldots,\Yvec_p$, where each $\Yvec_h,\,h\in\{1,\ldots,p\}$ is a $n$ dimensional vector consisting of $n$ observations for the $h^\mathrm{th}$ response variable; i.e., $\Yvec_h=\left(Y_{1h},\ldots,Y_{nh}\right)^T$. Let $\Xvec_1, \ldots,\Xvec_q$ be the $q$ covariates (e.g., observed external covariates or prognostic factors), which influence the association between the response variables. These associations are modeled  using a varying sparsity framework (Section~\ref{varying_coeff_form}), which enables variable selection of the response variables, for the parent set of a given node, and also aids in the variable selection of covariates, influencing the edges between a given node and its parents. Each $\Xvec_k$, $k\in\{1,\ldots,q\}$ is also a $n$ dimensional vector over the same observations for the $k^\mathrm{th}$ covariate; i.e., $\Xvec_k=\left(X_{1k},\ldots,X_{nk}\right)^T$. We define the quantile graph for  the $i^\mathrm{th}$ observation at quantile level $\tau\in(0,1)$ as $\mathcal{Q}\mathcal{G}^{(\tau)}_i=\left(\Vvec_i,\,\Evec_i^{(\tau)}\right)$ where the set of vertices $\Vvec_i$ represent $Y_{ih},\, h\in\{1,\ldots,p\}$ and the set of directed edges $\Evec_i^{(\tau)}$ between $Y_{ih}$ and $Y_{ij}$ for $h\neq j$ arise from the conditional quantile dependence of $Y_{ih}$ on $Y_{ij}$, at quantile level $\tau\in(0,1)$. With a slight abuse of notation, $\mathcal{Q}\mathcal{G}^{(\tau)}_i$ can be thought as an adjacency matrix, whose entries are either 1 or 0, based on presence or absence of the corresponding edge. Since in many applications such as protein networks it would be hard to interpret an edge having opposite directions in different DAGs, $\mathcal{Q}\mathcal{G}^{(\tau)}_i$, for $i=1,\dots,n$, we impose the following acyclic condition on the union of quantile graphs of all observations.
\begin{condition}
\label{union_DAG_Condition}
Let $ \mathcal{Q}\mathcal{G}^{(\tau)}_u=\bigcup_{i=1}^n \mathcal{Q}\mathcal{G}^{(\tau)}_i$ denote the union of quantile graphs at quantile $\tau$. We restrict $\mathcal{Q}\mathcal{G}^{(\tau)}_u$ to be a DAG.
\end{condition}
Under Condition~\ref{union_DAG_Condition}, each $\mathcal{Q}\mathcal{G}^{(\tau)}_i$ for $i\in\{1,\ldots,n\}$ is necessarily a DAG, since their union is a DAG. In other words, given an edge $Y_{ih}\leftarrow Y_{ij}$ $\forall\, h,\, j\in\{1,\ldots,p\},\,h\neq j$ at quantile level $\tau$ for some $i$, the edge $Y_{i'h}\rightarrow Y_{i'j}$ does not exist for any $i'$. This restriction is sensible in our motivating biological application of  patient-specific protein--protein interaction networks in lung cancer, where the direction of the edges arises from some shared biological phenomena. For example, in the
analysis of gene expression data, the strengths of regulatory links may vary across individuals, but in
general they do not change direction~\citep{wang2020high}. The union-DAG condition~\eqref{union_DAG_Condition} also makes qDAGx computationally tractable because one does not have to check the acyclicity of $\mathcal{Q}\mathcal{G}^{(\tau)}_i$ for all $i$; checking $\mathcal{Q}\mathcal{G}^{(\tau)}_u$ would suffice. 

Since all  $\mathcal{Q}\mathcal{G}^{(\tau)}_i$s are DAGs, there exists an ordering of the nodes such that for every directed edge $Y_{ih}\leftarrow Y_{ij}$, $Y_{ih}$ appears before $Y_{ij}$ in the ordering. For a given edge $Y_{ih}\leftarrow Y_{ij}$, the node $Y_{ih}$ is called the child and $Y_{ij}$ the parent. The collection of all the parents of $Y_{ih}$ is denoted by $pa_i(h)$. Denoting $\Xvec_{i\cdot}=(X_{i1}, \ldots, X_{iq})$, we write the varying sparsity model for conditional quantile of $Y_{ih}$ at quantile level $\tau\in(0,1)$, denoted as $Q_{Y_{ih}}(\tau \mid \cdot)$, as follows:
\begin{equation}
\label{patient_specific_model}
     Q_{Y_{ih}}(\tau\mid Y_{ij},\,\Xvec_{i\cdot}) =  \beta_{h0}^{(\tau)}(\Xvec_{i\cdot}) + \sum_{j\in pa_i(h)} Y_{ij}\beta_{hj}^{(\tau)}(\Xvec_{i\cdot}),
\end{equation}
where $\beta_{h0}^{(\tau)}(\cdot)\text{ and }\beta_{hj}^{(\tau)}(\cdot)$ are the coefficients whose functional forms remain the same for all $i$ at a given $\tau$. However, their values change depending on the covariates $\Xvec_{i\cdot}$, which are observation specific. We say there is an edge $Y_{ih}\leftarrow Y_{ij}$ at quantile level $\tau$, if $\beta_{hj}^{(\tau)}(\Xvec_{i\cdot})\neq 0$, or in other words, the association between $Y_{ih}$ and $Y_{ij}$, at quantile level $\tau$, is quantified by $\beta_{hj}^{(\tau)}(\Xvec_{i\cdot})$.
\subsection{Working likelihood of the proposed model}
\label{section_qDAGX_working_likelihood}

Minimizing the `check loss'~\citep{koenker1978regression} gives the optimal values for the coefficients in the varying sparsity model of~\eqref{patient_specific_model}, which further leads to a working likelihood \citep{yang2016posterior}. With the observation-specific model in~\eqref{patient_specific_model}, the check loss function can be obtained as, 
\begin{equation}
\label{loss_function}
    L(\tau) = \sum_{i=1}^n\sum_{h=1}^{p}\psi_\tau\left(Y_{ih} -  \beta_{h0}^{(\tau)}(\Xvec_{i\cdot})  - \sum_{j\in pa_{i}(h)} Y_{ij}\beta_{hj}^{(\tau)}(\Xvec_{i\cdot})\right),
\end{equation}
where $\psi_\tau(x) = \tau x\Ind(x\geq 0) -(1-\tau)x\Ind(x<0)$. 
The check loss in \eqref{loss_function} can also be viewed as a sum of negative log-likelihoods of an additive noise model with independently distributed asymmetric Laplace noises $u_{ih}$ with density $f(u_{ih}\mid \tau) = \tau(1-\tau)\exp(-\psi_\tau(u_{ih}))$ where $  u_{ih} =   Y_{ih} -  \beta_{h0}^{(\tau)}(\Xvec_{i\cdot}) -  \sum_{j\in pa_{i}(h)} Y_{ij}\beta_{hj}^{(\tau)}(\Xvec_{i\cdot})$. Hence, the working likelihood is:
\begin{equation}
\label{DAG_working_likelihood}
\begin{split}
 \pi(\Yvec \mid \Xvec,\,\tau,\,\betavec^{(\tau)}) & = \prod_{i=1}^n\prod_{h=1}^{p}\tau(1-\tau)\exp\left\{-\psi_\tau\left( Y_{ih} -  \beta_{h0}^{(\tau)}(\Xvec_{i\cdot})  - \sum_{j\neq h} Y_{ij}\beta_{hj}^{(\tau)}(\Xvec_{i\cdot}) \right)\right\}\\
 & \times \Ind\left(\mathcal{Q}\mathcal{G}^{(\tau)}_u \text{ is a DAG}\right).
\end{split}
\end{equation}
The index $j$ runs over $j\neq h$ in~\eqref{DAG_working_likelihood} instead of $j\in pa_{i}(h)$ as in~\eqref{loss_function} but they are equivalent because 
$\beta_{hj}^{(\tau)}(\Xmat_{i\cdot})=0$ for $j\notin pa_i(h)$ and hence, $\underset{j\neq h}{\sum}Y_{ij}\beta_{hj}^{(\tau)}(\Xmat_{i\cdot}) = \underset{j\in pa_{i}(h)}{\sum}Y_{ij}\beta_{hj}^{(\tau)}(\Xmat_{i\cdot})$. As a special case of~\eqref{DAG_working_likelihood}, the working likelihood of a quantile-DAG in the case of a given DAG can be written in a factorized form, $\pi\left(\Yvec \mid \Xvec,\,\tau,\,\betavec\right) = \prod_{h=1}^{p} \pi(\Yvec_h\mid \Xvec,\, pa(h),\, \tau)$. Without loss of generality, assuming the ordering of nodes as $\Yvec_1,\ldots,\Yvec_p$ in the case of known ordering, we can write the node conditional working likelihood at a node $\Yvec_h$ as,    
\begin{equation}
\label{factorised_likelihood_known_ordering}
     \pi\left(\Yvec_h\mid\cdot\right)  = \prod_{i=1}^n\tau(1-\tau)\exp\left\{-\psi_\tau\left( Y_{ih} -  \beta_{h0}^{(\tau)}(\Xvec_{i\cdot})  - \sum_{j=h+1}^{p} Y_{ij}\beta_{hj}^{(\tau)}(\Xvec_{i\cdot}) \right)\right\}.
\end{equation}
Our motivation to introduce the special case of known ordering in~\eqref{factorised_likelihood_known_ordering} is simply to use it as an `oracle' and to compare the  estimation results under this case, against the results of qDAGx, which does not assume a known ordering, a fact we reiterate. Though estimating coefficients $\beta_{hj}^{(\tau)}(\Xvec_{i\cdot})$ by optimizing \eqref{factorised_likelihood_known_ordering} is much simpler than doing the same in~\eqref{DAG_working_likelihood}, it may also be a bit unrealistic as the structure and the ordering of the DAG are unknown in most practical applications. Hence, in this paper, our focus is on estimating the model specified in~\eqref{DAG_working_likelihood}. Relevant numerical results are deferred to Section~\ref{simulation_results}.   

For a given response variable $\Yvec_h$, fitting the observation-specific model \eqref{patient_specific_model} for all observations, in the case of a given DAG, is equivalent to estimating the coefficients in \eqref{factorised_likelihood_known_ordering}. Learning the parameters of $\mathcal{Q}\mathcal{G}^{(\tau)}_i,\,\forall\,i$  consists of working independently with $p-1$ such densities:  $\pi\left(\Yvec_1\mid\cdot\right),\ldots,\pi\left(\Yvec_{p-1}\mid\cdot\right)$, where $\Yvec_{p}$ is understood to have no parents. However, when the DAG is unknown, a naive independent estimation of coefficients is no longer possible as the factorization depends on the unknown DAG and therefore the coefficients and DAG structure need to be estimated jointly. As the quantile-DAG,  $\mathcal{Q}\mathcal{G}_i^{(\tau)}$, is inferred based on the estimates of $\beta_{hj}^{(\tau)}(\Xvec_{i\cdot})$, establishing identifiability of the inferred DAGs remains a  key challenge. This is because, for a given observation $i$, it is unclear whether there are different quantile-DAG structures, all resulting in the same likelihood. We address this issue by proving structural identifiability of our model in Section~\ref{identifiability_section}.

\subsection{Functional form of the coefficients in varying sparsity model}
\label{varying_coeff_form}
As mentioned in the beginning of Section~\ref{qDAGx_master_intro_section}, the varying sparsity framework enables selection of the response variables and also aids in the selection of covariates influencing an edge between the response variables. In this section, we formally introduce the functional form of coefficients in this framework, $\beta_{hj}^{(\tau)}(\cdot)$, including the intercept terms $\beta_{h0}^{(\tau)}(\cdot)$ in~\eqref{patient_specific_model}, and call it the quantile conditional independence function (QCIF). For $h=1,\ldots,p$ and $j\in \{pa(h) \cup 0\}$, we write QCIF as a product of a smooth function and a hard thresholding operator as follows:  
\begin{align}
    \label{beta_hj}
        \beta_{hj}^{(\tau)}(\Xvec_{i\cdot}) = \theta_{hj}^{(\tau)}(\Xvec_{i\cdot})\cdot \Ind\big(|\theta_{hj}^{(\tau)}(\Xvec_{i\cdot})|>t_{hj}\big),\quad& \theta_{hj}^{(\tau)}(\Xvec_{i\cdot})=  \sum_{k=1}^q f_{hjk}^{(\tau)}(X_{ik}).
\end{align}
We choose to model  $f_{hjk}^{(\tau)}(\cdot)$ as a smooth semi-parametric function. These functions make the edge strengths,  $\theta_{hj}^{(\tau)}(\cdot)$, vary smoothly according to the values of  covariates. Specifically, we model  $f_{hjk}^{(\tau)}(\cdot)$ with cubic B-Splines \citep{cox1972numerical, de1972calculating}. Based on the recommendation by~\citet[Supp. material, Section A]{ni2019bayesian}, we consider B-splines with a large number of bases, $B$. The hard thresholding operator in \eqref{beta_hj} enables variable selection of the response variables, with the edge $Y_{ih}\leftarrow Y_{ij}$ being present if  $|\theta_{hj}^{(\tau)}(\Xvec_{i\cdot})|$ is larger than a certain threshold $t_{hj}$. Variable selection of the covariates is aided by the ability to compute the probability of posterior inclusion (Supplementary Section~\ref{Bayesian_estimation_and_variable_selection}), of the effect of $\Xvec_k,\,k\in\{1,\ldots,q\}$, (captured by the spline coefficients used to model $f_{hjk}^{(\tau)}(\cdot)$) on the edges $Y_{ih}\leftarrow Y_{ij},\,\forall\,i$.

\subsection{Prior formulation}
\label{prior_formulation_overall}
In \eqref{beta_hj}, we denote, $ \{\thetavec_{hj}^{(\tau)}(\Xmat_{i\cdot})\} =(\theta_{hj}^{(\tau)}(\Xvec_{1\cdot}),\ldots,\theta_{hj}^{(\tau)}(\Xvec_{n\cdot}))^T$, a $n$ dimensional vector containing $\theta_{hj}^{(\tau)}(\Xvec_{i\cdot})$ for all observations, $i=1,\ldots,n$. Let $\Xmat$ denote the set of all covariates. Using B-Splines to model  $\{f_{hjk}^{(\tau)}(X_{ik})\}=(f_{hjk}^{(\tau)}(X_{1k}),\ldots,f_{hjk}^{(\tau)}(X_{nk}))^T$, we can write, $\{f_{hjk}^{(\tau)}(X_{ik})\} = \widetilde{\Xvec_k}\alphavec_{hjk}$, where $\widetilde{\Xvec_k}$ is the design matrix of size $n\times B$, with $B$ denoting the number of basis functions, corresponding to the spline coefficients $\alphavec_{hjk}$. In order to avoid overfitting, we use penalized splines \citep{eilers1996flexible, lang2004bayesian} penalizing the second order differences between adjacent spline coefficients. The penalty can be written as $\alphavec_{hjk}^T \bSigma \alphavec_{hjk}$ where  $\bSigma$ is a fixed singular, positive semi-definite matrix. It is clear that the penalty is the negative of the logarithm (up-to additive constants) of a normal density,  $\alphavec_{hjk} \sim \mathcal{N}(0, \lambda_s \bSigma^{-})$, where $\lambda_s$ is the smoothness parameter (analogous to tuning parameter in ridge regression) and $\bSigma^{-}$ is a generalized matrix inverse of $\bSigma$ \citep{ruppert2003semiparametric}.

For computational convenience, we reparameterize $\alphavec_{hjk}$ following \citet{scheipl2012spike}, to obtain a proper normal density, which is proportional to the improper prior density on $\widetilde{\Xvec_k}\alphavec_{hjk}$. First, taking the spectral decomposition of covariance of $\widetilde{\Xvec_k}\alphavec_{hjk}$, we observe that,
\begin{equation*}
    \mathrm{cov}(\widetilde{\Xvec_k}\alphavec_{hjk}) = \lambda_s\widetilde{\Xvec_k}\bSigma^{-}\widetilde{\Xvec_k}^{T} = \lambda_s\begin{bmatrix}
    \Umat_k & *
    \end{bmatrix}\begin{bmatrix}
    \Dmat_k & 0\\
    0 & 0
    \end{bmatrix}\begin{bmatrix}
    \Umat_k & *
    \end{bmatrix}^{T},
\end{equation*}
where $\Umat_k$ is an orthonormal matrix of eigenvectors corresponding to the positive eigenvalues in the diagonal matrix $\Dmat_k$; with the other eigenvectors suppressed by *.  Note that the spline coefficients that correspond to the linear and constant terms are in the null space of $\bSigma$ and hence are not penalized. Defining $\widetilde{\Xmat_k}^{*} = \Umat_k \Dmat_k^{1/2}$ and $\alphavec_{hjk}^* \sim \mathcal{N}(0, \sigma_{hj}^2\Imat_{B_k^*})$, we can see $\widetilde{\Xmat_k}^{*}\alphavec_{hjk}^*$ admits a proper normal density, which is proportional to the improper prior density on $\widetilde{\Xvec_k}\alphavec_{hjk}$. Further, $\widetilde{\Xmat_k}^{*}$ is a matrix of dimension $n\times B_k^*$,\, $\alphavec_{hjk}^*$ is a vector of dimension $B_k^* \times 1$, and $\Imat_{B_k^*}$ is an identity matrix of dimension $B_k^* \times B_k^*$. For computational reasons, we set $B_k^*$ as the number of eigenvalues in $\Dmat_k$, sorted in decreasing order,  which explain at least $99.5\%$ of variability in $\{f_{hjk}^{(\tau)}(X_{ik})\}$. This dimension reduction approach is also followed by \citet{ni2019bayesian} and leads to efficient computation, as the number of spline coefficients is $\mathcal{O}(\text{max}(B_1^*,\ldots,B_q^*)qp)$ instead of $\mathcal{O}(Bqp)$, and $\text{max}(B_1^*,\ldots,B_q^*)$ is typically much smaller than $B$. In the numerical experiments described in Section~\ref{simulation_results}, the observed values for $B_k^*$ are $5$ or $6$, whereas we have $B=20$. Thus, reparameterizing the penalty on nonlinear spline coefficients and with unpenalized linear and intercept terms, we can write 
\begin{equation*}
     \{f_{hjk}^{(\tau)}(X_{ik})\} = \mu_{hjk}\bm{1}_{n} + \widetilde{\Xmat_k}^{*}\alphavec_{hjk}^* + \Xvec_k\alpha_{hjk}^0,  
\end{equation*}
where $\alpha_{hjk}^0,\,\mu_{hjk}$ are the unpenalized linear coefficient and intercept respectively and $\bm{1}_{n}$ is the $n$ dimensional unit vector. With this simplification of $\{f_{hjk}^{(\tau)}(X_{ik})\}$,
\begin{equation}
\label{theta_formulation}
     \{\thetavec_{hj}^{(\tau)}(\Xmat_{i\cdot})\} = \sum_{k=1}^q \{f_{hjk}^{(\tau)}(X_{ik})\} = \mu_{hj}\bm{1}_{n} + \sum_{k=1}^q \widetilde{\Xmat_k}^{*}\alphavec_{hjk}^* + \sum_{k=1}^q \Xvec_k\alpha_{hjk}^0,
\end{equation}
where we absorb all the intercept terms into $\mu_{hj}$. This Gaussian prior on $\alphavec_{hjk}^{*}$, together with our choice of priors for $\mu_{hj}\text{ and } \alpha^0_{hjk}$, almost completes the prior specification, pending one further detail. A Gaussian prior on the spline coefficients yields a ridge penalty, which does not provide strong enough shrinkage to zero in a sparse regime. The remedy is to use a Gaussian \emph{scale mixture} prior instead of just a Gaussian to simultaneously achieve a sharper pull towards zero and heavier tails in the marginal prior. Global-local shrinkage priors such as the horseshoe \citep{carvalho2010horseshoe} are known to outperform ridge estimates under sparse settings \citep{polson2012local, bhadra2021horseshoe, polson2019bayesian} and they all belong to the family of Gaussian scale mixtures. Thus, to achieve stronger shrinkage, we model the prior on $\alphavec_{hjk}^*$ using an appropriate Gaussian scale mixture, specified in the next subsection.  Computational performances under such priors have been studied by \citet{gelman2008using} and  \citet{scheipl2012spike}, who empirically observed good performances in terms of MCMC mixing.
\subsubsection{Induced prior on the structure of quantile-DAG and model fitting}
\label{text_pxHS_priors}
The \emph{parameter expanded normal mixtures of inverse gamma (peNMIG)} prior \citep{gelman2008using, scheipl2012spike}, a special case of scale mixtures of normals, has been used on the spline coefficients (analogous to $\mu_{hj},\, \alphavec_{hjk}^*\text{ and }\alpha_{hjk}^0$) in varying sparsity Bayesian quantile regression \citep{QUANTICO} and in DAG inference using node conditional  varying sparsity model \citep{ni2019bayesian}. Motivated by the peNMIG, we propose the  \emph{parameter expanded horseshoe (pxHS)} prior in this paper; where the prior on the mixing scale variable is half Cauchy. Choosing priors for $\mu_{hj},\, \alphavec_{hjk}^*,\, \alpha_{hjk}^0$ in \eqref{theta_formulation} and $t_{hj}$ in \eqref{beta_hj}, completes the prior specification for all parameters as follows:
\begin{align}
\label{pxHS_priors}
    \mathrm{pxHS: }& \begin{cases}
         \alphavec_{hjk}^* = \eta_{hjk}\xivec_{hjk}\,,\,\eta_{hjk} \sim \mathcal{N}(0, T_{hj}^2 L_{hjk}^2)\,,\,
          & 
          \xivec_{hjk} = \left(\xi_{hjk}^{(1)},\ldots,\, \xi_{hjk}^{(B_k^*)}\right)^T,\\
        \xi_{hjk}^{(l)}  \sim \mathcal{N}(m_{hjk}^{(l)}, \sigma_m^2)\,,\,\text{for }l\in\{1,\ldots,B_k^*\},\; & m_{hjk}^{(l)}  \sim 0.5\cdot \delta_1 (m_{hjk}^{(l)}) +0.5\cdot \delta_{-1} (m_{hjk}^{(l)}),\\
         T_{hj} \sim \mathcal{C}^+(0,1),\;& L_{hjk} \sim \mathcal{C}^+(0,1),
    \end{cases}\nonumber\\
    \alpha_{hjk}^0 &\sim \text{pxHS prior analogous to } \alphavec_{hjk}^*,\\
    \mu_{hj}&\sim \mathcal{N}(0,\,\sigma_\mu^2)\text{ and }t_{hj}\sim\text{Gamma}(\text{shape = }a,\,\text{rate = }b),\, 1\leq h, j\leq p . \nonumber
\end{align}
Taking a closer look at the pxHS prior in \eqref{pxHS_priors}, we see that the nonlinear spline coefficient $\alphavec_{hjk}^*$ is written as a product of two random variables, a scalar $\eta_{hjk}$ and a vector $\xivec_{hjk}$. 
The scalar $\eta_{hjk}$ is sampled from a horseshoe density \citep{carvalho2010horseshoe} that belongs to the class of global-local shrinkage priors which control sparsity at two levels: global and local. Here $T_{hj}$ serves as the global scale parameter, regulating sparsity among nonlinear spline coefficients $\alphavec_{hj1}^*,\ldots, \alphavec_{hjq}^*$ and $L_{hjk}$ serves as the local scale parameter, regulating sparsity among $B_k^*$ number of entries in each $\alphavec_{hjk}^*$. Further, the horseshoe prior on $\eta_{hjk}$ is scaled to all $B_k^*$ spline coefficients by a mixture-normal random variable $\xivec_{hjk}$. As elements of $\xivec_{hjk}$ are concentrated around $1\text{ and }-1$, they discourage small values of $\alphavec_{hjk}^*$. In the gamma prior on the thresholds, $t_{hj}$, the values $a, b$ are chosen such that the prior mean is equal to the average expected edge strength and $\sigma_\mu$ in the prior on $\mu_{hj}$ is fixed. Note that the thresholds $t_{hj}$ remain the same for all observations $i\in\{1,\ldots,n\}$.  Additional details on the hyperparameters, posterior sampling via MCMC and variable selection procedures are outlined in Supplementary Sections~\ref{complete_Gibbs_samplers}--\ref{Bayesian_estimation_and_variable_selection}. 
\section{Theoretical Properties}
\label{theorectical_properties_qDAGx}
The main theoretical properties of our approach are outlined in this section. We begin by showing the identifiability of qDAGx. We then demonstrate the `non-local' property of our prior that aids estimation and inference in a sparse regime \citep{rossell2017nonlocal}. We conclude by establishing the posterior consistency of the node conditional fitted densities, lending strong theoretical support to the proposed methodology.

\subsection{Identifiability of qDAGx}
\label{identifiability_section}
DAGs are generally only identifiable up to Markov equivalence classes. Within each Markov equivalence class, DAGs encode the same conditional independence relationships. The practical implication is that one cannot hope to identify the true data generating DAG even with an infinite amount of data. While this is generally the case, we prove that the proposed quantile-DAG, $\mathcal{Q}\mathcal{G}^{(\tau)}$, inferred by qDAGx,  is identifiable at any quantile level and at any covariate value. That is, there do not exist two distinct DAGs that lead to the same likelihood function. The following theorem formalizes the claim.  
\begin{theorem}
\label{identifiability_theorem}
There do not exist $\betavec^{(\tau)'}\neq\betavec^{(\tau)}$ such that $\pi(\bY \mid \bX,\,\tau,\,\betavec^{(\tau)})\equiv \pi(\bY \mid \bX,\,\tau,\,\betavec^{(\tau)'})$. 
\end{theorem}
The proof is provided in Supplementary Section~\ref{proof_of_identifiability}. Note that since $\mathcal{Q}\mathcal{G}^{(\tau)}$ is induced by $\betavec^{(\tau)}$, Theorem \ref{identifiability_theorem} implies that no two observation-specific quantile-DAGs have the same likelihood. Therefore, it is possible to identify the true data generating $\mathcal{Q}\mathcal{G}^{(\tau)}$. 
\subsection{Non-localness of the marginal prior on QCIFs}
\label{non_local_proof}
Non-local priors, introduced by~\citet{johnson2010use}, are a class of priors which have zero mass at the null value of the parameter and are proven to \emph{ameliorate the imbalance}~\citep{johnson2010use} in rates of convergence and accumulation of evidence in favor of the true hypothesis. Such priors have been used recently for high-dimensional estimation~\citep{rossell2017nonlocal, shin2018scalable}, with attractive results in terms of lower estimation errors and false discovery rates, when compared to penalized likelihoods approaches such as the lasso or SCAD. In this section, we prove that the marginal prior on QCIFs is a mixture of point mass at zero and a non-local prior, which have been termed  `mass nonlocal' priors~\citep{shi2019model} and have been found to inherit the advantages of a spike-slab prior, where the spike is a point mass and the slab component is non-local. Considering the functional form of $\theta_{hj}^{(\tau)}(\Xvec_{i\cdot})$ in \eqref{theta_formulation} and the priors defined on the constant, linear and non-linear spline coefficients, we need a convolution of horseshoe priors with a normal prior, followed by a truncation, to obtain the marginal prior on the QCIFs,   $\beta_{hj}^{(\tau)}(\Xvec_{i\cdot})$. Unfortunately, the convolution of horseshoe priors is not analytically tractable. Hence, we prove the non-localness result in a simpler case of $\Xvec$ being a scalar quantity and $\sigma_m^2\rightarrow 0$ (in~\eqref{pxHS_priors}). With this simplification, we can write the QCIFs as  $\beta_{hj}^{(\tau)}$. As this result is true for any $\beta_{hj}^{(\tau)}$, we suppress the edge-specific and quantile-specific notation and state the result for a general QCIF $\beta$ where $\beta = \theta\cdot\Ind(|\theta|>t)$ and a horseshoe prior is imposed on $\theta$. Denoting the $\mathrm{Gamma}(a,\,b)$ prior on threshold as $\mathcal{P}_t(t)$ and horseshoe prior on $\theta$ as $\mathcal{P}_\theta(\cdot)$, our non-localness result on the marginal prior on $\beta$, is as follows.
\begin{lemma}
\label{NLP_lemma}
The marginal prior $\mathcal{P}(\beta)$ is a mixture of point mass at 0, $\delta_0(\cdot)$, and a non-local prior $\pi(\beta)$, $\mathcal{P}(\beta)=w\delta_0(\beta)+(1-w)\pi(\beta)$ 
where $w=E_t\{Pr(|\theta|\leq t\mid t)\}$ and
 $\pi(\beta)=\mathcal{P}_\theta(\beta)\frac{Pr(t<|\beta|\mid \beta)}{E_t\{Pr(|\theta|> t\mid t)\}}\rightarrow 0$ as $\beta\rightarrow 0$.
\end{lemma}
\begin{proof}
Following \citet[Proposition 1]{ni2019bayesian}, it is straightforward that $\mathcal{P}(\beta)$ is a mixture distribution as stated in the statement of this lemma. What remains to prove is that $\pi(\beta)\to 0$ as $\beta\to0$. We will first show $\mathcal{P}_\theta(\beta)Pr(t<|\beta|\mid \beta)\rightarrow 0$ as $\beta\rightarrow 0$ and then prove that $[E_t\{Pr(|\theta|> t\mid t)\}]^{-1}$ is bounded. From the properties of horseshoe density \citep{carvalho2010horseshoe}, we know that when $\beta\rightarrow 0$, $\mathcal{P}_\theta(\beta)\approx -\log\beta$. And from the CDF of gamma distribution, we have $Pr(t<|\beta|\mid \beta)\sim O(\beta^{a})$ when $\beta\rightarrow 0$. Thus, $\mathcal{P}_\theta(\beta)Pr(t<|\beta|\mid \beta)\rightarrow 0$ when $\beta\rightarrow 0$. We prove that  $[E_t\{Pr(|\theta|> t\mid t)\}]^{-1}$ is bounded in the supplementary Section.~\ref{predictive_consistency_proof}; for explicit bounds see the displays \eqref{UB_of_horseshoe_required} and \eqref{LB_of_horseshoe_required}. 
\end{proof}
The result of Lemma~\ref{NLP_lemma} explains the motivation behind the hard thresholding operator for automatic variable selection of the response variables. This choice of prior also aids in sparse quantile-DAG discovery, as mass nonlocal priors shrink small effects to zero and allow the selection of only meaningful edges $\Yvec_h\leftarrow\Yvec_j$ in individualized quantile-DAGs.
\subsection{Posterior consistency of the node conditional fitted densities}
\label{sec_predictive_consistency}
In this section, we prove the consistency of the node-conditional fitted density for a given node at any quantile level $\tau\in(0,\,1)$. For reasons mentioned in Section~\ref{non_local_proof}, we prove the consistency result when $\Xvec$ is a unit matrix i.e., the case of no covariates and $\sigma_m^2\rightarrow 0$ (in~\eqref{pxHS_priors}). Following \eqref{patient_specific_model}, the simplified model in the case of unknown ordering at a given node $\bY_h$ can be written as, 
\begin{equation}
\label{simplified_general_model}
     Q_{Y_{ih}}(\tau\mid Y_{ij}) =  \beta_{h0}^{(\tau)}+ \sum_{j\neq h} Y_{ij}\beta_{hj}^{(\tau)},
\end{equation}
where $\beta_{hj}^{(\tau)} = \theta_{hj}^{(\tau)}\Ind(|\beta_{hj}^{(\tau)}|>t_{h})$ and independent horseshoe priors are imposed on $\theta_{hj}^{(\tau)}$. The horseshoe prior on $\theta_{hj}^{(\tau)}$ with a global scale parameter $c$ (fixed) can be written as a half Cauchy scale mixture of normal density as follows,
\begin{equation}
    \label{simplified_prior_equation}
    \theta_{hj}^{(\tau)} \mid c,u_{hj} \sim \mathcal{N}(0,\, u_{hj}^2 c^2),\quad u_{hj}\sim \mathcal{C}^+(0,1),\,\quad c>0.
\end{equation}
Let $\bgamma_h=\{\gamma_{hj}\}$ be a vector of zeros and ones denoting the absence and presence of $\bY_{j},\, j\neq h$ in the model~\eqref{simplified_general_model} respectively i.e., if $\gamma_{hj}=0$, then $\beta_{hj}^{(\tau)}=0$ and $\bY_{j},\, j\neq h$ is not in the model and vice-versa. Also, let $\betavec_{\bgamma_{h}}$ denote the vector of non-zero QCIFs specified by the model $\bgamma_{h}$. Let $p_n$ denote the number of response variables as a function of $n$ i.e.,  $\bY_1,\ldots,\bY_{p_n}$ and let $r_n = p_n\pi_n$ where $\gamma_{hj} \sim \text{Bernoulli}(\pi_n)$. Here, $\pi_n$ is analogous to $w$ in Lemma~\ref{NLP_lemma}, and as $w$ is a function of shape and rate parameters of the gamma prior on the threshold, $a,\,b$ and the global scale parameter $c$. Let $\epsilon_n$ denote a sequence of positive numbers decreasing to zero and $1\prec n\epsilon_n^2$ where  $u_n \prec v_n$ means $\lim_{n\to\infty}u_n/v_n = 0$. Under the true data generating model as given in \eqref{simplified_general_model}, define $\Delta_h(r_n) = \underset{|\bgamma| = r_n}{\mathrm{inf}}\sum_{j\notin\bgamma,\,j\neq h}|\beta_{hj}^{*(\tau)}|$, where $\betavec_{h}^*=\{\beta_{hj}^{*(\tau)}\}$ denotes the vector of true QCIFs. Here $\Delta_h(r_n)$ captures the residual effect i.e., sum of absolute values of all true QCIFs which are absent in the model $\bgamma_{h}$. Now restricting the total model size $\sum_{j\neq h}\gamma_{hj}\leq \bar{r}_n$, we introduce Conditions A1--A7, which are required to prove the consistency of the node conditional fitted density. 
\vspace{1em}

    \begin{tabular}{lll}
        $\text{A1. }~\bar{r}_n\log p_n \prec n\epsilon_n^2,$ & $\text{A2. }~\bar{r}_n \log\left(\frac{1}{\epsilon_n^2}\right)\prec n\epsilon_n^2,$ &  $\text{A3. }~1\leq r_n\leq \bar{r}_n\leq p_n,$ \\
        $\text{A4. }~\sum_{j\neq h}|\beta_{hj}^{*(\tau)}|<\infty$ and $|\beta_{hj}^{*(\tau)}|>0,$ &  $\text{A5. }~1\prec r_n\prec p_n < n^\alpha,\, \alpha >0,$ & $\text{A6. }~p_n\Delta_h(r_n)\prec \epsilon_n^2,$
    \end{tabular}
    
    \begin{tabular}{l}
         $\text{A7. }~a>3,\, a\in\mathrm{Z}^+$ and $b,\,c>0$  such that $c^2b<3/2.$ 
    \end{tabular}
    \vspace{1em}
    
Conditions A1--A6 can also be found in \citet{jiang2007bayesian}, which are used to establish consistency of the fitted densities in Bayesian high dimensional variable selection. As described in~\citet{guha2020quantile}, Condition A1 is used to bound the metric entropy of the carefully chosen sieve in the space of prior densities. Condition A2 is required to ensure sufficient prior probability in the Kullback-Leibler (KL) neighborhood of the true model. Conditions A3 and A5 give the growth rates of model dimensions as a function of $n$, and Condition A4 guarantees that all the true QCIFs are finite and the absolute values are bounded away from zero. Condition A6 ensures that the residual effect is small and Condition A7 is used to derive bounds on $E_t[Pr(|\theta|>t\mid t)]$ (Lemma~\ref{NLP_lemma}), which can further be used to derive bounds for the prior concentration rate of KL $\epsilon_n^2$ neighborhoods and the prior probability of the complement of the chosen sieve. With these conditions in place, let  $H_h$ denote the Hellinger distance between the node conditional fitted density and the true node conditional density of $\Yvec_h$, given by: 
\begin{equation*}
    H_h^2 = {\int_{\bY}^{} \left(\sqrt{\pi(\bY_h\mid \bY_{j\neq h};\,\bgamma,\,\betavec_{\bgamma_{h}})} - \sqrt{\pi^*(\bY_h\mid \bY_{j\neq h};\,\betavec^*_h)} \right)^2\pi^*(\bY_{j\neq h})d\Ymat},
\end{equation*}
where $\pi^*(\cdot)$ denotes the true conditional density. Let $\pi_\tau(\cdot\mid\Ymat)$ denote the posterior probability  under the model and prior, as defined in \eqref{simplified_general_model} and \eqref{simplified_prior_equation} respectively. Then, we have the following theorem.
\begin{theorem}
\label{predective_consistency}
Let  $\mathrm{sup}_jE|\bY_j| = M^*<\infty$. Then under Conditions A1--A7, for some $c_1'>0$ and for $n^\delta\prec p_n\prec n^\alpha$, $\alpha>\delta>0$, and under the true data generating model \eqref{simplified_general_model} for some given quantile level $\tau$, the following holds for $n \to \infty$:
\begin{equation*}
    P^*\left\{\pi_\tau\left(H_h \leq \epsilon_n \mid \bY\right)>1-\exp(-c_1'n\epsilon_n^2)\right\}\rightarrow 1,
\end{equation*}
where $P^*(\cdot)$ denotes the probability under true data generating density. 
\end{theorem}
In the statement of the above theorem we have $n^\delta\prec p_n\prec n^\alpha$; $\alpha>\delta>0$. Now in particular if $\bar{r}_n\prec n^b$ with $b=\mathrm{min}\{\xi, \delta\}$ with $\xi \in (0,1)$, we have the rate of convergence $\epsilon_n = n^{-(1-\xi)/2}$ and the decaying rate is of the order $\exp\left(-n^\xi\right)$. This decay rate follows from Remark 1 of \citet{guha2020quantile}. The proof of Theorem~\ref{predective_consistency} is provided in the Supplementary Section~\ref{predictive_consistency_proof}.

The posterior consistency established in Theorem~\ref{predective_consistency} is interesting because we operate with a `pseudo' likelihood based on a loss function. This is because the quantile based conditional distributions may not correspond to a valid joint distribution \citep{guha2020quantile}. Nevertheless, we are able to establish posterior consistency of the node conditional fitted densities. Related ideas on posterior concentration properties for loss-based or Gibbs posterior inference for quantile regression have recently been explored by \citet{bhattacharya2022gibbs}. 
\section{Numerical Experiments}
\label{simulation_results}
The goal of this section is to compare the quantile-DAGs inferred on synthetic data sets by qDAGx and the oracle, where the quantile-DAGs are estimated with a known ordering of the nodes. We do this comparison for 12 settings of $(n,\,p,\,q)$, where $n\in\{100,250\},\, p\in\{25,50,100\}\text{ and } q\in\{2,5\}$. Before going into the details of simulation results, we present how we generate the synthetic data, in the following five steps.
\begin{enumerate}
\setlength\itemsep{-0.1em}
    \item[(a)] Assuming the true ordering of nodes to be $\{\Yvec_1,\ldots,\Yvec_p\}$, we randomly select $\text{max}\Big\{1, \lfloor\frac{p-h}{5}\rfloor\Big\}$ number of nodes in  $\{\Yvec_{h+1},\ldots,\Yvec_p\}$, as the parents of the node $\Yvec_h$, thus keeping the true DAG 80\% sparse. 
    
    \item[(b)] The covariates $\Xvec_1, \ldots, \Xvec_q$ are generated as $n$ i.i.d samples from a multivariate normal, $\mathcal{N}(0,\Imat_q)$. We set each $\theta_{hj}^{(\tau)}(\cdot)$ as a function of a subset of $\Xvec$, thus introducing varying sparsity. To do this, we choose $q^*
    \in\{0,\,1,\,2\}$ when $q=2$ and  $q^*\in\{0,\,1,\,2,\,3\}$ when $q=5$, and randomly choose $q^*$ number of covariates from $\Xvec$ in computing each  $\theta_{hj}^{(\tau)}(\cdot)$. For different values of $q^*$, the functional forms of  $\theta_{hj}^{(\tau)}(\cdot)$ are as follows:
    \begin{enumerate}
        \item[$(\mathrm{i})$] For $q^*=0$, \; $\{\theta_{hj}^{(\tau)}(\Xmat_{i\cdot})\} =(1+\tau^2)\bm{1}_{n}\,$, where $\bm{1}_{n}\text{ is the unit vector of dimension }n$. 
        \item[$(\mathrm{ii})$] For $q^*=1$, \; $\{\theta_{hj}^{(\tau)}(\Xmat_{i\cdot})\} =\Xvec_{k_{1}}^2 + \log ((1+\tau^2)\bm{1}_{n}$), where $k_1\text{ is randomly chosen from }\{1,\ldots,q\}$.
        \item[$(\mathrm{iii})$] For $q^*=2$, \; $\{\theta_{hj}^{(\tau)}(\Xmat_{i\cdot})\} =\Xvec_{k_{1}}^2 + \log ((1+\tau^2)\bm{1}_{n}) + \exp(\Xvec_{k_{2}})$, where $k_1,\,k_2$ are distinct and randomly chosen from $\{1,\ldots,q\}$.
        \item[$(\mathrm{iv})$] For $q^*=3$, \; $\{\theta_{hj}^{(\tau)}(\Xmat_{i\cdot})\} =\Xvec_{k_{1}}^2 + \log ((1+\tau^2)\bm{1}_{n}) + \exp(\Xvec_{k_{2}}) + \log|\Xvec_{k_{3}}|$, where $k_1,k_2, k_3$ are distinct and randomly chosen from $\{1,\ldots,q\}$.
    \end{enumerate}
    \item[(c)] We fix all thresholds$(t_{hj})$ equal to 0.5 when $q=2$ and equal to 1 when $q=5$ and set $\{\beta_{hj}^{(\tau)}(\Xmat_{i\cdot})\} = \{\theta_{hj}^{(\tau)}(\Xvec_{i\cdot})\}\cdot\Ind(|\{\theta_{hj}^{(\tau)}(\Xvec_{i\cdot\}})|>t_{hj})$. With the parents of each nodes selected in step (a) and from the functional form of coefficients in step (b), we plug in $n$ i.i.d samples of  $\tau\sim\mathcal{U}(0,1)$ in our observation specific model~\eqref{patient_specific_model}, to generate $n$ samples of $\Yvec_h$, noting that the quantile and distribution functions are inverses of each other.  
    
    \item[(d)] After data for all nodes $\Yvec_h,\,h\in\{1,\ldots,p\}$ are generated, we compute the values of $\theta_{hj}^{(\tau)}(\cdot),\,\beta_{hj}^{(\tau)}(\cdot)$ and $Q_{Y_{ih}}(\tau \mid \cdot)$ at nine equally spaced quantile levels $\tau\in\{0.1,\ldots,0.9\}$. We store those these values as matrices $\thetavec_h^{\tau,\text{ true}},\,\betavec_h^{\tau,\text{ true}}\text{ and }\Qmat_{\Yvec_{h},\tau}^{\text{true}}$ respectively, at $\tau\in\{0.1,\ldots,0.9\}$ and use them to compute the estimation norms (Supplementary Equations~\eqref{estimation_norms} and~\eqref{MSE_known_ordering}).
    
    \item[(e)] Given a setting of $(n,p,q)$, we repeat steps (b)--(d) to simulate 25 different data sets. We compute mean and standard deviation of the results over these data sets. 
\end{enumerate}
We estimate the model parameters using MCMC sampling (Supplementary Section~\ref{complete_Gibbs_samplers}) and infer the quantile-DAG structure in three scenarios:  $(a)$ qDAGx with known ordering (denoted as $\mathrm{qDAGx}_{0}$ and referred to as the \emph{oracle} in this paper) $(b)$ qGADx with unknown ordering (denoted as qDAGx itself), and $(c)$ qDAGx with a misspecified ordering (denoted as $\mathrm{qDAGx}_{\mathrm{m}}$). The last case is considered to check the robustness of our procedure i.e., by checking for lower false positive rates in variable selection of the response variables and covariates. Just as in the case of known ordering, misspecified ordering also admits a factorization of likelihood (a wrong one in fact) and enjoys parallel and independent estimation of parameters. To mis-specify the order, we use Kendall's rank correlation coefficient or Kendall's $\mathrm{T}$ \citep{kendall1938new} and use two levels of misspecification, with $\mathrm{T} = 0.25\text{ and }0.5$. For example, a misspecified ordering with $\mathrm{T}=0.25$, is a randomly chosen permutation of $\Yvec_1,\ldots,\Yvec_p$ such that the rank correlation between the permuted and true orderings is 0.25.  

At each setting of $(n,\,p,\,q)$, we infer the quantile-DAGs at nine equally spaced quantile levels $\tau\in\{0.1,\ldots,0.9\}$. The inferred quantile-DAGs from the three scenarios are compared against each other using nine different performance metrics:  true positive rate ($\mathrm{TPR}_{\mathrm{Y}}^\tau,\,\mathrm{TPR}_{\mathrm{X}}^\tau$), false positive rate ($\mathrm{FPR}_{\mathrm{Y}}^\tau,\,\mathrm{FPR}_{\mathrm{X}}^\tau$) and area under receiver operating characteristic curve ($\mathrm{AUC}_{\mathrm{Y}}^\tau,\,\mathrm{AUC}_{\mathrm{X}}^\tau$) of variable selection of (response variables, covariates) respectively; estimation norms $ \Delta_F\betavec^\tau,\,\Delta_F\thetavec^\tau$, and lastly the adjusted mean squared error in quantile estimation,  $\text{MSE}^\tau$. A detailed description of how these metrics are computed, is given in the Supplementary Section~\ref{extra_simulation_results}. For two representative settings of $(n,\,p,\,q)$, we compare the quantile-DAGs inferred in all the three scenarios in Fig.~\ref{all_3_models_comparison}. Similar results were observed in all the other numerical experiments performed; the details of which are deferred to Supplementary Section~\ref{extra_simulation_results}. 
\begin{figure}[!bp]
\centering
    \begin{minipage}{\textwidth}
        \includegraphics[width=0.97\textwidth]{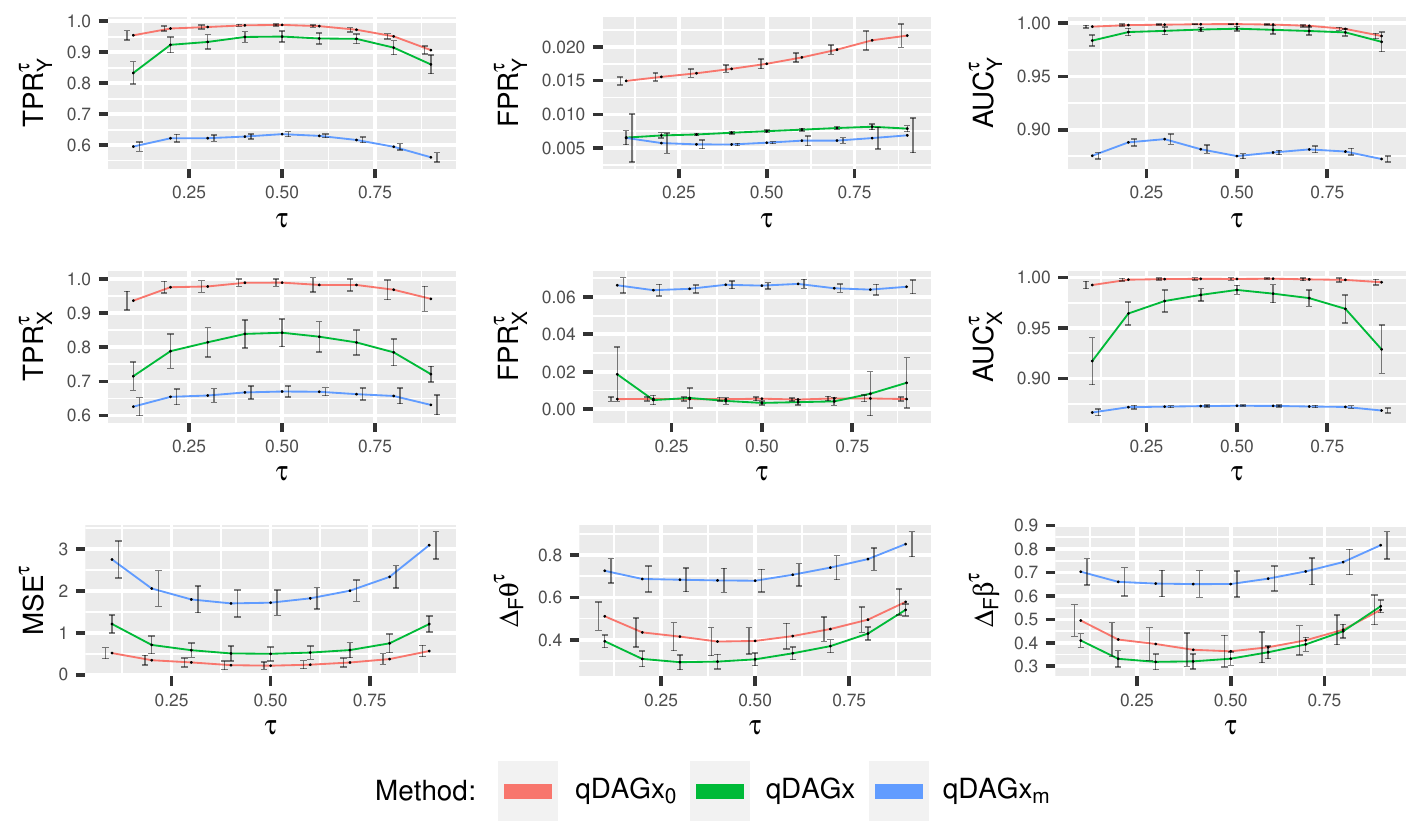}
        \Scaption{\label{all_3_models_comparison_a} $p=25, q=5, n=250$. Kendall's' $\mathrm{T}$ for the misspecified sequence is 0.5}
    \end{minipage}\hfill
    \medskip
    \begin{minipage}{\textwidth}
        \includegraphics[width=0.97\textwidth]{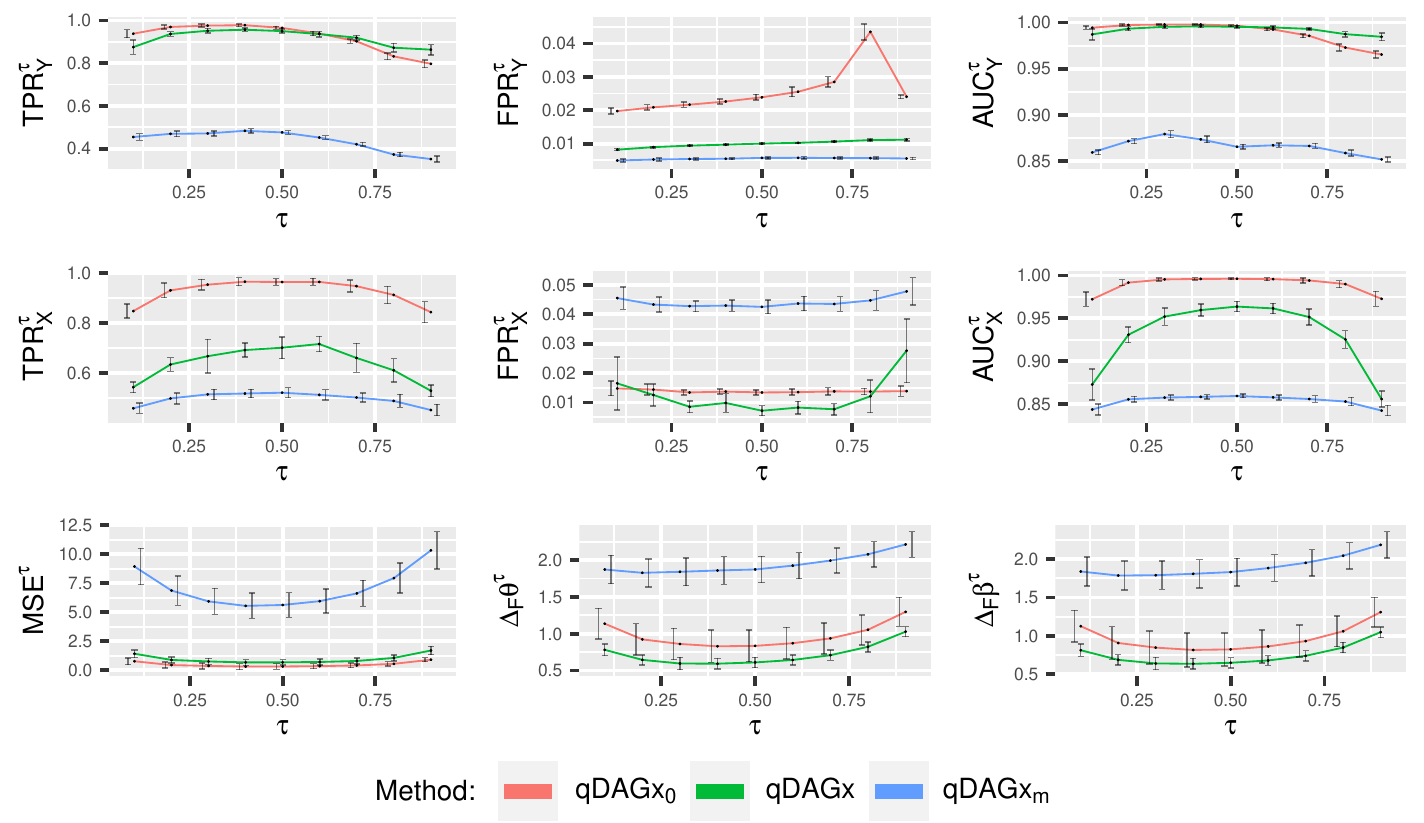}
        \Scaption{\label{all_3_models_comparison_b} $p=50, q=2, n=250$. Kendall's' $\mathrm{T}$ for the misspecified sequence is 0.25}
    \end{minipage}\hfill
    \caption{\label{all_3_models_comparison}Simulation results for two representative settings comparing the nine performance metrics between the quantile-DAG estimates of $\mathrm{qDAGx}_{0}$, qDAGx and $\mathrm{qDAGx}_\mathrm{m}$.}
\end{figure}

 In each of the nine panels in Fig.~\ref{all_3_models_comparison}\phantom{ }\ref{all_3_models_comparison_a} and Fig.~\ref{all_3_models_comparison}\phantom{ }\ref{all_3_models_comparison_b}, the mean of the corresponding performance metrics from 25 data sets is plotted as a black dot, with error bars covering one standard deviation on either sides of the mean.  From Fig.~\ref{all_3_models_comparison}\phantom{ }\ref{all_3_models_comparison_a}, we can see that the performance of qDAGx is very competitive to the oracle $\mathrm{qDAGx}_{0}$, in terms of variable selection at the response variables. In fact, it has true positive rates and area under the curve matching the oracle. In terms of variable selection of covariates, the true positive rate of qDAGx is lower that that of the oracle, but much higher than the misspecified model $\mathrm{qDAGx}_\mathrm{m}$; and the area under the ROC curve is comparable to the oracle. Similar trends are observed  in  Fig.~\ref{all_3_models_comparison}\phantom{ }\ref{all_3_models_comparison_b}. Coming to estimation norms, we can see that in both Fig.~\ref{all_3_models_comparison}\phantom{ }\ref{all_3_models_comparison_a} and Fig.~\ref{all_3_models_comparison}\phantom{ }\ref{all_3_models_comparison_b}, qDAGx has the lowest $\Delta_F\betavec^\tau\text{ and }\Delta_F\thetavec^\tau$. The adjusted mean squared error in quantile estimation, $\mathrm{MSE}^\tau$, is also very competitive to the oracle. As can be expected, the model with the misspecified ordering of nodes, performs the worst in all performance indicators. With these simulation results (including results in Supplementary Section~\ref{extra_simulation_results}), we establish that qDAGx produces results which are often comparable to the oracle, sometimes outperforming it. Though there are no exact competing procedures that do two level (response variable and covariate) quantile graphical modeling as qDAGx, we compare the results of qDAGx with that of the quantile graphs inferred by lasso penalized quantile regression \citep{wu2009variable}, `lasso-QR'. The results are presented in Supplementary Section~\ref{supplement_comparision_Lasso_QR}, where it can be seen that lasso-QR has a comparable performance only in variable selection of response variables and performs poorly in all other metrics.

\section{Application of qDAGx for Individualized Inference in Lung Cancer}
\label{real_data_application_qDAGx}

We demonstrate the practical applicability of qDAGx  by inferring individualized (patient-specific) protein--protein interaction networks in two sub-types of non-small cell lung cancer: lung adenocarcinoma (LUAD) and lung squamous cell carcinoma (LUSC). It is well-established that cancer is caused  by complex changes at molecular and genetic levels, which primarily arise from aberrations in protein or gene regulatory networks or signaling pathways  \citep{boehm2011towards, vogelstein2004cancer}. Emphasis on individual-specific, or genotypic drivers of gene networks or pathways, in case of lung cancer, is well laid out in~\citet{beer2002gene, shen2019precision}, indicating a need for individualized inference on these networks. The data we examine comes from The Cancer Genome Atlas (TCGA) consortium \citep{weinstein2013cancer}, which has collated  proteomic, genomic, and clinical  data from over 7700 patients across 32 different cancer types.    From the TCGA database, we focus our analyses on a subset of patients with LUAD and  LUSC \citep{cancer2012comprehensive, cancer2014comprehensive, campbell2016distinct}. We consider the proteomic data measured using Reverse Phase Protein Array (RPPA) technology, which is further streamlined and processed by \citet{ha2018personalized} into functional pathways. Pathways are sets of proteins that are collectively responsible for cellular functions such as apoptosis, cell cycle, DNA damage response, that moderate different oncological processes. By studying the network topology of these pathways, our goal is a deeper understanding of the key \emph{individualized} functional changes that may induce carcinogenesis \citep{bandyopadhyay2010rewiring}. 

In our study, we consider messenger-RNA  (mRNA) and methylation as  two observed external covariates (i.e, $q=2$), and set, $\Xmat=(\Xvec_1,\, \Xvec_2)=(\text{mRNA, methylation})$. Our goal is to integrate information from transcriptomic (mRNA) and DNA (methylation) data to decipher proteomic networks that might be disrupted during the oncogenic process. The scientific motivation stems from the fact that mRNA carries instructions from the DNA into ribosome and hence plays a crucial role in the protein synthesis, whereas {methylation is known to play an important role in the regulation of protein--protein interactions, transcriptions and other biological processes}~\citep[Summary, Chapter 11]{raju2019co}, both with clinical significance in cancer ~\citep{hargrove1989role, phillips2008role}. We use protein expression data for  $p=67$ proteins across 12 pathways, for $n=306$ patients with LUAD and $n=278$ patients with LUSC.  We infer patient-specific quantile-DAGs at $\tau\in\{0.1,\ldots,0.9\}$ using qDAGx. The names of these 67 proteins are presented in Supplementary Section~\ref{extra_real_data_analysis}, Table~\ref{protein_number_mapping}. 

\subsection{Biological interpretations at the  individual level}
\label{real_data_individual_inference}

Once the posterior samples for qDAGx are available for every patient via MCMC for $\tau\in\{0.1,\ldots,0.9\}$, we choose that posterior sample (quantile-DAG) for every patient at a quantile level $\tau$, which is closest to the posterior mean of the corresponding patient. Quantile-DAGs of a randomly chosen patient with LUAD, at $\tau=\{0.1,\,0.5,\,0.9\}$ is presented in Figure~\ref{rand_patient_LUAD_tau_159}. It can be seen from the figure that the set of nodes with high out-degrees i.e., the set of influential proteins for each pathway, changes across different quantile levels. This reemphasizes our objective to study quantile-DAGs, instead of Gaussian DAGs, to achieve a better understanding of the extremal dependence.
\begin{figure}[!htb]
    \centering
    \begin{minipage}{0.33\textwidth}
        \centering
        \includegraphics[scale=0.69]{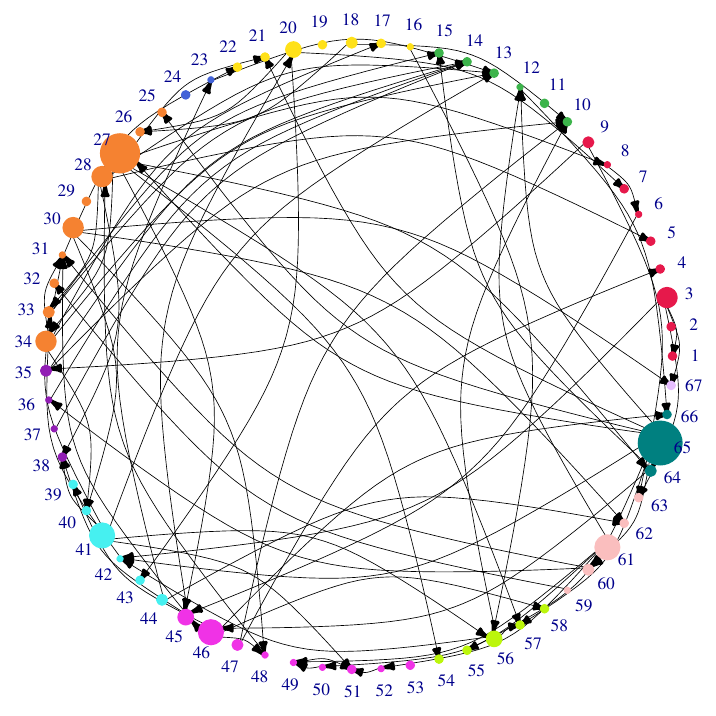}
        \Scaption{\label{rand_patient_LUAD_tau_1}}
    \end{minipage}\hfill
    \begin{minipage}{0.33\textwidth}
        \centering
         \includegraphics[scale =0.69]{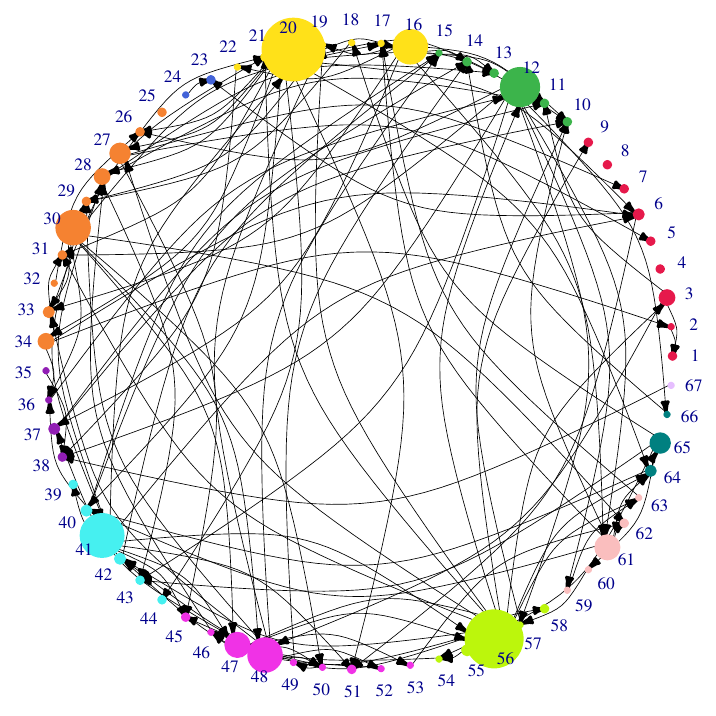}
        \Scaption{\label{rand_patient_LUAD_tau_5}}
    \end{minipage}
        \begin{minipage}{0.33\textwidth}
        \centering
        \includegraphics[scale=0.69]{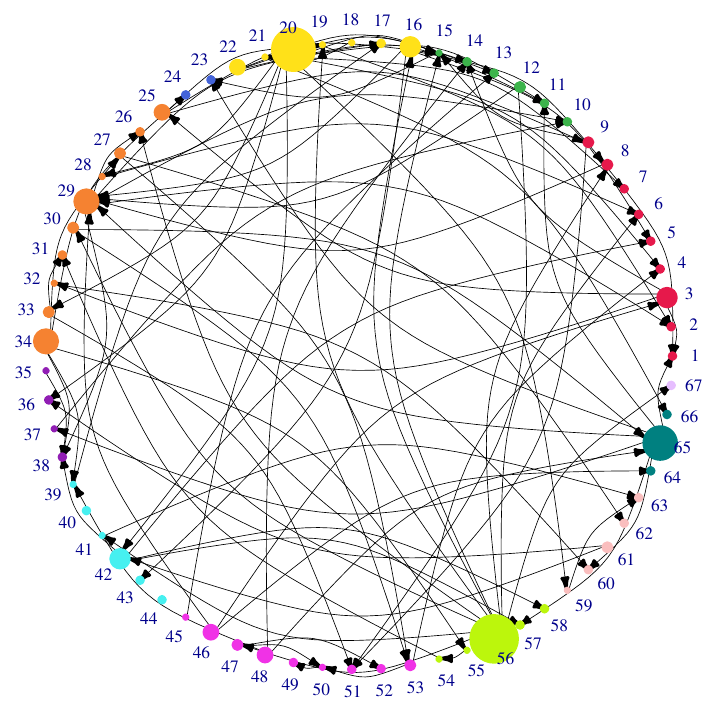}
        \Scaption{\label{rand_patient_LUAD_tau_9}}
    \end{minipage}\hfill
    \caption{\label{rand_patient_LUAD_tau_159}Panels (a), (b) and (c) are visualizations of quantile-DAGs at $\tau=0.1,\,0.5\text{ and }0.9$ respectively, for a randomly selected patient with LUAD. Node sizes are proportional to out-degree of nodes and nodes are colored according to the pathway to which they belong. The map between node colors and pathway names is given in Supplementary Table~\ref{pathway_color_mapping}. Note that there are some proteins which belong to multiple pathways; such proteins are just assigned to one of the pathways for the sake of clear visualization.}
\end{figure}

While we infer quantile-DAGs for each patient, for further interpretation and illustration, we present results for those edges in the inferred DAGs which are present in at least 50\% of the patients and across five different quantile levels. These protein--protein associations for LUAD and LUSC are presented in Table~\ref{LUAD_LUSC_min_50p_5q}. In what follows, we discuss the main implications of our findings as summarized in Table~\ref{LUAD_LUSC_min_50p_5q}, and their connections with previous literature. For the edge, \edge{BAK1}{BID}, \citet{sarosiek2013bid} have identified that in cancers including the lung, \texttt{BID} preferentially activates \texttt{BAK1}, with implications on chemotherapy response. Studies performed by \citet{sasaki2006egfr, li2012lung, zhang2013prognostic} are among the many works in lung cancer, which have studied mutations in \texttt{EGFR} and \texttt{ERBB2}; and found biological evidence for the effect of \texttt{ERBB2} (also known as \texttt{HRE2}) on \texttt{EGFR} (\edge{EGFR}{ERBB2}). As for the edge \edge{PCNA}{CHEK1}, \citet{yang2008chk1} have shown that  \texttt{CHEK1} (also known as \texttt{CHK1}) regulates the DNA damage-induced degradation of the protein \texttt{PCNA} (via Ubiquitination), which facilitates the continuous replication of damaged DNA.  \citet{becker2018flap} corroborate the same finding in several other cancers. For the edge \edge{BAD}{AKTS1}, the evidence is indirect. \citet{wang2012pi3k} have identified that mRNAs of \texttt{PTEN}, \texttt{ATKS1} and \texttt{BAD} are significantly down-regulated in lung cancer cells, which are resistant towards large-dose and short-period radiation therapies. \citet{knuppel2018fk506} have
studied the up and down regulations of \texttt{COL6A1} and \texttt{CAV1} in presence of
a protein which regulates
lung fibroblast migration, which in turn is known to enhance the migration of cancer cells~\citep{camci2015fibroblasts}. Further, a visual representation of prevalence, using two representative edges \edge{CAV1}{COL6A1} and \edge{MYH11}{COL6A1} from Table~\ref{LUAD_LUSC_min_50p_5q} across both the cancers, is presented in Figure~\ref{COL6A1_both_cancers}. It can be seen from the figure that the prevalence of the edges is not uniform across different quantile levels, strengthening our argument for a quantile based inference.  An overall summary of the effect of mRNA and methylation on the protein--protein interactions is presented in Table~\ref{percentage_second_level_covar_cancer}. It is instructive to note from the table that a strikingly high percentage of the edges is influenced by both mRNA and methylation; and further, methylation alone influences about twice the number of edges when compared to mRNA alone.

    \begin{table}[h]
    \caption{\label{LUAD_LUSC_min_50p_5q} Directed edges in quantile-DAG estimates which are present in at least 50\% of patients and across five out of nine quantile levels, $\tau\in\{0.1,\ldots,0.9\}$. Common edges in LUAD and LUSC are in bold.\\}
    \centering
    \resizebox{\textwidth}{!}{
    \begin{tabular}{|ccc|ccc|}
    \hline
    \multicolumn{3}{|c|}{Lung adenocarcinoma (LUAD)} & \multicolumn{3}{c|}{Lung squamous cell carcinoma (LUSC)} \\ \hline
    \textbf{\edge{BAK1}{BID}}  &    \edge{BAD}{ATK1S1}         & \edge{BID}{ERBB3}    & \textbf{\edge{BAK1}{BID}}    & \edge{AKT1, AKT2, AKT3}{AKT1S1} & \edge{CAV1}{PGR}     \\
    \textbf{\edge{CAV1}{COL6A1}} & \textbf{\edge{EGFR}{ERBB2}}         & \edge{GAPDH}{CDH2}   & \textbf{\edge{CAV1}{COL6A1}} & \textbf{\edge{EGFR}{ERBB2}}           & \edge{CCNB1}{COL6A1} \\
    \edge{JUN}{ERBB3}   & \textbf{\edge{MAPK1, MAPK3}{MAP2K1}} & \textbf{\edge{MYH11}{COL6A1}} & \edge{MTOR}{PGR}    & \textbf{\edge{MAPK1, MAPK3}{MAP2K1}}     & \textbf{\edge{MYH11}{COL6A1}} \\
    \edge{PCNA}{CHEK1}  & \textbf{\edge{RPS6KB1}{PGR} }        &                                         & \edge{MYH11}{FOXM1} &  \textbf{\edge{RPS6KB1}{PGR}}              &  \edge{RAD51}{PGR}\\
    \hline
    \end{tabular}
    }
    \end{table}

\begin{figure}[!t]
    \centering
    \includegraphics[width = \textwidth]{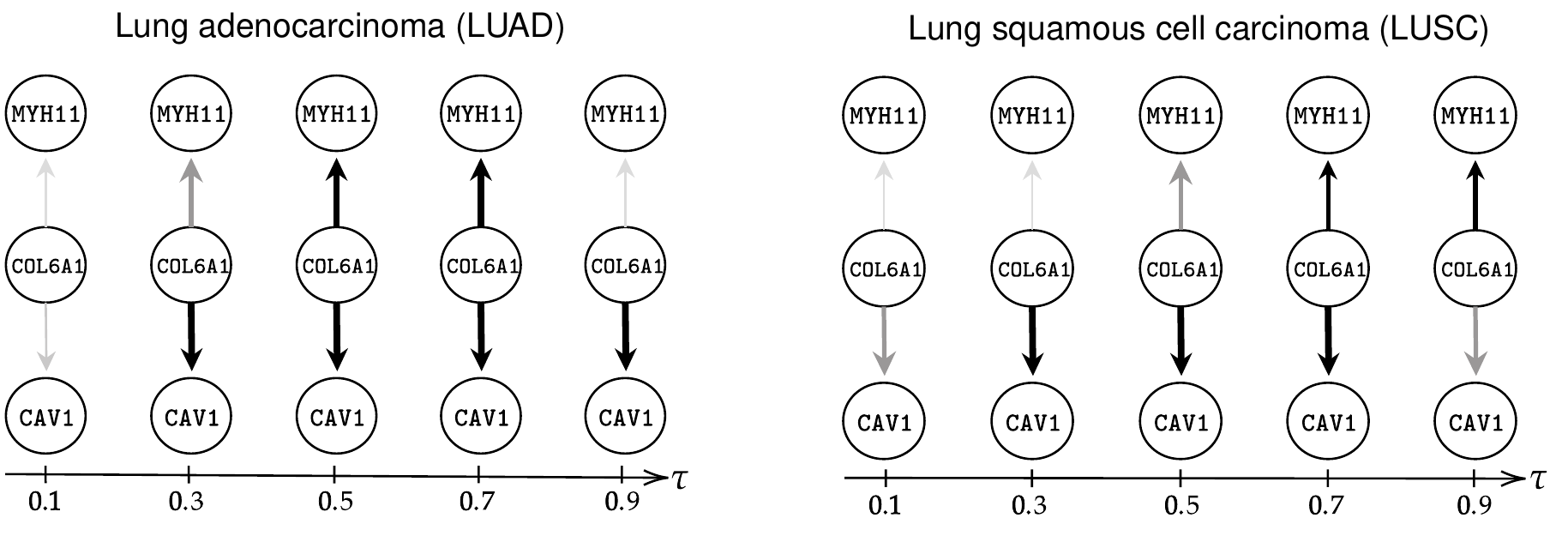}
    \caption{\label{COL6A1_both_cancers}Prevalence of \edge{CAV1}{COL6A1} and \edge{MYH11}{COL6A1} in LUAD and LUSC. Boldness of the edge is proportional to the number of patients in whom the edge was inferred at the specific quantile level $\tau$.}
\end{figure}

    \begin{table}[!htb]
    \caption{\label{percentage_second_level_covar_cancer} Mean (sd) for the percentage of edges influenced by  covariates (only mRNA, only methylation, both mRNA and methylation) in quantile-DAG estimates of all patients, across the quantiles $\tau\in\{0.1,\ldots,0.9\}$. The covariate $\bX_k\in\{\text{mRNA, methylation}\}$ is counted as an influence on the edge $\Yvec_h\leftarrow\Yvec_j$ if the pseudo probability of posterior inclusion (Supplementary Section~\ref{Bayesian_estimation_and_variable_selection}) is greater than 0.5.\\}
    \centering
    \begin{tabular}{|c|c|c|c|}
    \hline
         & only mRNA    & only methylation & both         \\
         \hline
    LUAD & 13.7 (0.78) & 28 (0.86)    & 58.3 (1.55) \\
    LUSC & 13.6 (0.66) & 28.1 (0.61)     & 58.3 (0.85)\\
    \hline
    \end{tabular}
    \end{table}

\subsection{Biological interpretations at the population level}
\label{real_data_population_inference}

To draw inferences at the population level, we aggregate the quantile-DAGs of all patients over posterior samples, at every quantile level, following \citet{chowdhury2020dagbagm}. The aggregated quantile-DAG is set as the mean of all the patient-specific quantile-DAGs, which in-turn are chosen as the closest DAG structure to the respective MCMC posterior means. The aggregated quantile-DAG at quantile level $\tau$ is denoted by $\Evec^{(\tau)}_{\text{LUAD}}$ and $\Evec^{(\tau)}_{\text{LUSC}}$ for the cancers LUAD and LUSC respectively. We present $\Evec^{(\tau)}_{\text{LUAD}}$ for $\tau\in\{0.1,\, 0.5,\, 0.9\}$ where the node size is proportional to in-degree of nodes, respectively in Fig.~\ref{LUAD_tau_1_5_9_results}~ \ref{LUAD_tau_1}, \ref{LUAD_tau_5} and \ref{LUAD_tau_9}.  In Fig.~\ref{LUAD_tau_1_5_9_results}\ref{LUAD_out_deg_tau_5},  $\Evec^{(0.5)}_{\text{LUAD}}$ is visualized when  the node size is proportional to out-degree of nodes, as it is easier to interpret. Similar aggregated quantile-DAGs for Lung squamous cell carcinoma (LUSC) are presented in Supplementary Fig.~\ref{LUSC_tau_1_5_9_results}.

\begin{figure}[bp!]
    \centering
    \begin{minipage}{0.5\textwidth}
        \centering
        \includegraphics[scale=1]{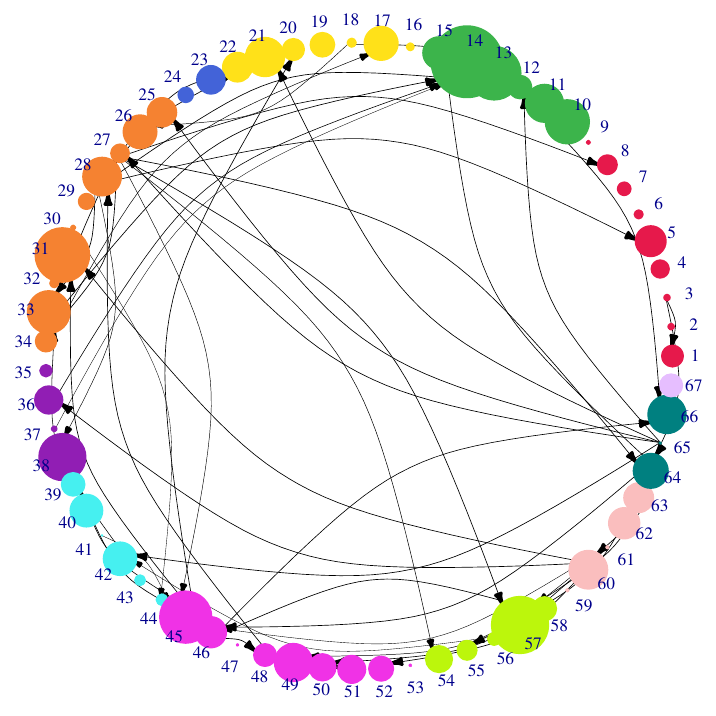}
        \Scaption{\label{LUAD_tau_1}}
    \end{minipage}\hfill
    \begin{minipage}{0.5\textwidth}
        \centering
         \includegraphics[scale =1]{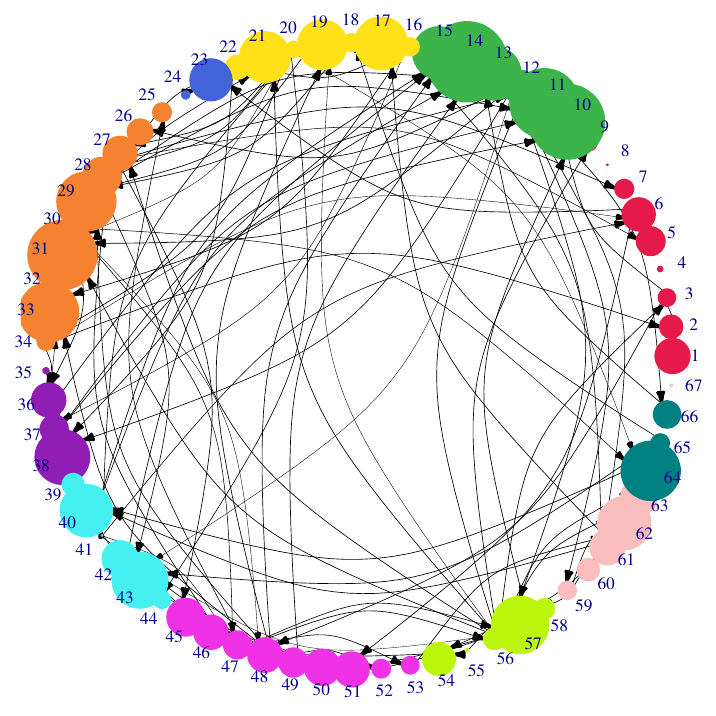}
        \Scaption{\label{LUAD_tau_5}}
    \end{minipage}
        \begin{minipage}{0.5\textwidth}
        \centering
        \includegraphics[scale=1]{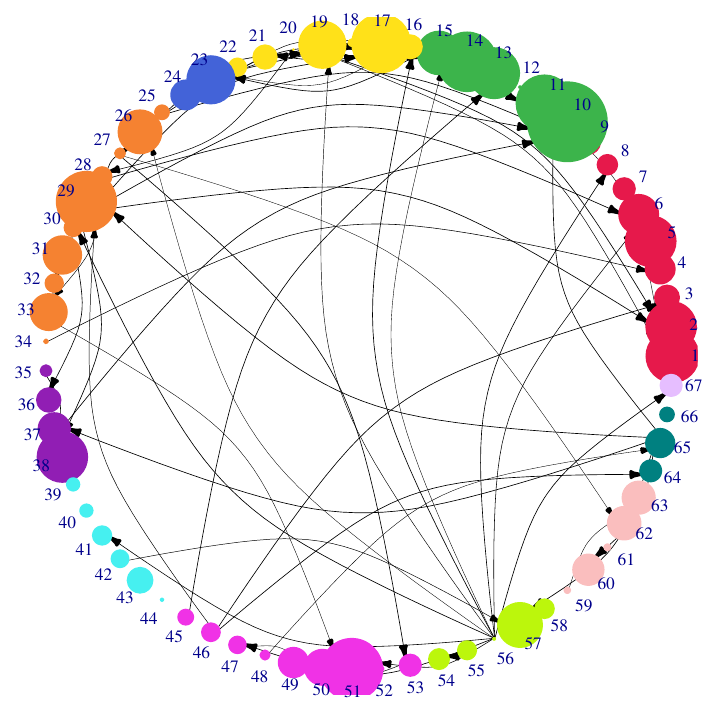}
        \Scaption{\label{LUAD_tau_9}}
    \end{minipage}\hfill
     \begin{minipage}{0.5\textwidth}
        \centering
        \includegraphics[scale=1]{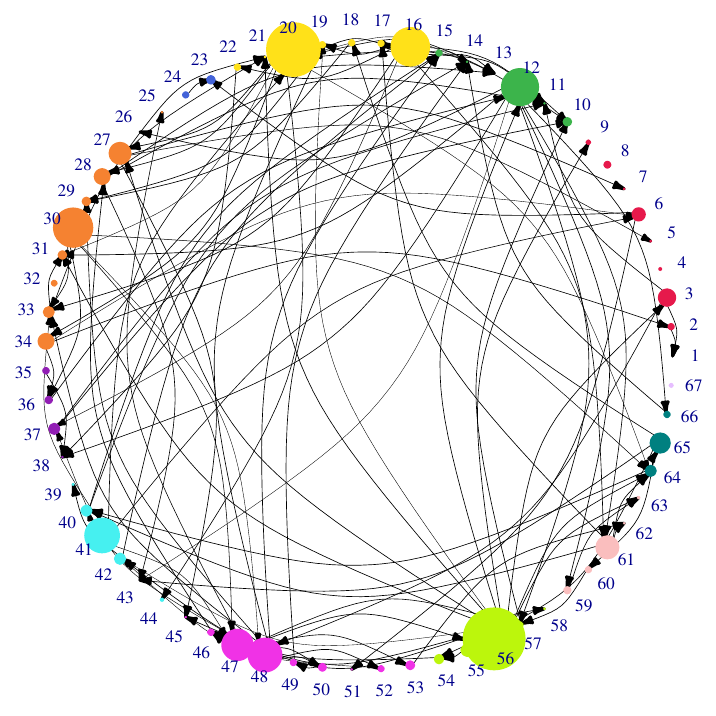}
        \Scaption{\label{LUAD_out_deg_tau_5}}
    \end{minipage}
    \caption{\label{LUAD_tau_1_5_9_results} Panels (a), (b), (c) show aggregated quantile-DAGs, $\Evec^{(\tau)}_{\text{LUAD}}$, for $\tau=0.1,\,0.5\text{ and }0.9$ respectively. Panel (d) shows $\Evec^{(0.5)}_{\text{LUAD}}\,$,  when the node size is proportional to out-degree of nodes. In all panels, nodes are colored according to the pathway to which they belong. The map between node colors and pathway names is given in Supplementary Table~\ref{pathway_color_mapping}. Note that there are some proteins which belong to multiple pathways; such proteins are just assigned to one of the pathways for the sake of clear visualization.}
\end{figure}

To identify the hub nodes, we rank the proteins in a descending order by their in-degrees with respect to the aggregated quantile-DAGs estimated at $\tau\in\{0.1\ldots,0.9\}$; and then pick those proteins which appear in the top-3 positions in at least four different aggregated quantile-DAGs. Doing so, the identified `key' proteins (in decreasing order) in LUAD are: \texttt{CAV1, GAPDH, SHC1} and in LUSC are: \texttt{GAPDH, CAV1, SHC1}. We summarize the main biological implications of these findings. The controversial role played by the protein \texttt{CAV1} in LUAD and LUSC is discussed in \citet{fu2017different}; where it is mentioned that \texttt{CAV1} plays a  tumor-inhibitory role in LUAD but a tumor-promoting role in LUSC. Relation between proliferation of cancer cell growth and the protein \texttt{GAPDH} in LUAD and LUSC is discussed in~\citet{tokunaga1987enhanced} and~\citet{hao2015elevated} respectively. Up-regulation of the protein \texttt{SHC1} in both the lung cancers was observed in a recent study conducted by~\citet{liang2021increased}.

Identifying the `key' proteins when the ranking is by out-degree instead of in-degree, we identify \texttt{ERBB3, CHEK1, MRE11A} as the hub nodes (in decreasing order) in LUAD and \texttt{ERBB3, RAB11A-RAB11B, PGR} in LUSC. The rank of \texttt{ERBB3} in the case of LUAD is interesting because, it ranked first in 7 out of 9 aggregated quantile-DAG estimates. Similar observation was made by \citet{sithanandam2003cell}, where the authors found that \texttt{ERBB3} was present at high levels in five of seven human lung adenocarcinoma cell lines examined. In a recent study by \citet{tan2022chek1}, the authors noted that \texttt{CHEK1} is a `hub' gene which corresponds to poor prognosis for lung adenocarcinoma. The role of \texttt{ERBB3} and its therapeutic targeting in various cancers including LUSC, is discussed by \citet{hafeez2020new}. Also in a general study of non-small cell lung cancers, which includes both LUAD and LUSC, \citet{dong2017rab11a} observed that \texttt{RAB11A} promotes the proliferation and invasion of cancer through the regulation of yes-associated protein (\texttt{YAP}). The role of progesterone receptors (\texttt{PGR}) in the tumourigenesis of non-small cell lung cancers has been recognized in \citet{su1996expression}. Overall, our findings agree with these previous studies, and unravel several other directed edges potentially of interest in lung cancer.
\section{Conclusions and Future Works}
\label{conclusion_qDAGx}
We introduce a novel quantile-DAG learning framework, termed qDAGx, which learns the association between variables at a given quantile with an unknown ordering of the  nodes. The key distinguishing feature of our approach is it provides \emph{individualized inference}, achieved via a varying sparsity framework, which is of interest in many modern applications, including precision medicine. Our framework also overcomes the major drawbacks of existing approaches such as the Gaussian DAG learning algorithms with known ordering, and models with a parametric likelihood. Our demonstration that the protein--protein interaction network varies at different quantile levels and across individuals in patients with LUAD and LUSC illustrate the usefulness of qDAGx in precision medicine. Our findings are corroborated by previous studies and potentially unravel new directional associations in studies of protein--protein interaction in lung cancer that are driven by upstream mRNA and methylation factors.

Several future directions could naturally follow from our work. First, theoretical guarantees of estimating the quantile-DAG structure remain  relatively unexplored. Recently, \citet{cao2019posterior} have proved the DAG estimation consistency in the posterior for Gaussian graphical models. Whether similar approaches are adequate for  proving the estimation consistency of quantile-DAGs remains an open question that is separate from the node conditional consistency results we establish. A second promising direction is to estimate the quantile-DAGs where strict conditions are imposed to preserve the increasing nature of quantiles \citep{NIPS2016_537de305, yang2017joint,das2018bayesian}. Recent works by \citet{yang2016posterior} and \citet{wu2021bayesian} have shown that the naive multiple quantile estimation does not give valid posterior summaries such as credible intervals or posterior means and have proposed adjustment methods. Individualized inference at multiple quantile levels while addressing the problem of \emph{quantile crossing} is challenging, both inferentially and computationally; and should be considered a future area of investigation.
\medskip
\begin{center}
{\large\bf SUPPLEMENTARY MATERIAL}
\end{center}
\medskip
(a) \emph{Supplementary Text:} contains proofs and additional results from simulations.\\ (b) \emph{Supplementary Code:} contains computer code archive along with a README file.
\setstretch{0.97}
\printbibliography
\clearpage\pagebreak\newpage
\doublespacing
\begin{center}
	{\LARGE{\bf Supplementary Material to\\
	{\it Bayesian Covariate-Dependent Quantile Directed
Acyclic Graphical Models for Individualized
Inference}}
	}
\end{center}
\setcounter{equation}{0}
\setcounter{page}{1}
\setcounter{table}{0}
\setcounter{section}{0}
\setcounter{subsection}{0}
\setcounter{figure}{0}
\renewcommand{\theequation}{S.\arabic{equation}}
\renewcommand{\thesection}{S.\arabic{section}}
\renewcommand{\thepage}{S.\arabic{page}}
\renewcommand{\thetable}{S.\arabic{table}}
\renewcommand{\thefigure}{S.\arabic{figure}}
\section{Complete Posterior Inference for qDAGx}
\label{complete_Gibbs_samplers}
In this section we detail the posterior updates of all parameters in qDAGx. A Gibbs update is performed when the full conditional is available, else the update proceeds using a random walk Metropolis step. The union-DAG condition~\eqref{union_DAG_Condition} needs to be obeyed at every MCMC iteration for qDAGx, while this condition is automatically met when sampling parameters in the case of known ordering (or oracle). In the sampling procedure, the pxHS prior on $\alpha_{hjk}^0$ mirrors the pxHS prior on $\alphavec_{hjk}^*$. So it is implicit that the updates of parameters in $\alpha_{hjk}^{0}$ mirror the updates of corresponding parameters in $\alphavec_{hjk}^*$. Also, we exploit the Inverse-Gamma (IG) parameter expansion technique \citep{makalic2015simple}, for half-Cauchy densities for the pxHS prior. With this technique, the half-Cauchy prior on $ T_{hj},\, L_{hjk}$ in $\alphavec_{hjk}^*$ can be written as,
\begin{equation*}
\begin{split}
    T_{hj}^2\mid c_{hj} \sim \text{IG}(1/2, 1/c_{hj}),&\,\, c_{hj}\sim \text{IG}(1/2,1),\\  
    L_{hjk}^2\mid \zeta_{hjk}\sim \text{IG}(1/2, 1/\zeta_{hjk}),&\,\, \zeta_{hjk}\sim \text{IG}(1/2,1).
\end{split}
\end{equation*}
The above hierarchy also helps us to compute the pseudo probabilities in variable selection of the covariates, which is described in the next subsection.  With the prior hierarchy complete, updates related to all parameters in $\beta_{hjk}^{(\tau)}(\Xvec_i)$ can be enumerated as follows:
\begin{enumerate}
    \item[(a)] Update $\eta_{hjk}$ by random walk metropolis with a normal proposal $\mathcal{N}(\eta_{hjk}^{(\text{old})},\sigma_\eta^2)$.
    \item[(b)] Update $T_{hj}^2$ by Gibbs, where $T_{hj}^2\sim IG\left(\frac{1+q}{2},\, \frac{1}{c_{hj}}+\frac{1}{2}\sum_{k=1}^{q}\left(\frac{\eta_{hjk}}{L_{hjk}}\right)^2\right)$.
    \item[(c)] Update $c_{hj}$ by Gibbs, where $c_{hj}\sim IG\left(1,\, 1+\frac{1}{T_{hj}^2}\right)$.
    \item[(d)] Update $L_{hjk}^2$ by Gibbs, where $L_{hjk}^2\sim IG\left(1,\, \frac{1}{\zeta_{hjk}}+\frac{1}{2}\left(\frac{\eta_{hjk}}{T_{hj}}\right)^2\right)$.
    \item[(e)] Update $\zeta_{hjk}$ by Gibbs, where $\zeta_{hjk}\sim IG\left(1,\, 1+\frac{1}{L_{hjk}^2}\right)$.
    \item[(f)] Update $\mu_{hj}$ by random walk metropolis with a normal proposal $\mathcal{N}(\mu_{hj}^{(\text{old})},\sigma_\mu^2)$.
    \item[(g)] Update $t_{hj}$ by random walk metropolis with a normal proposal $\mathcal{N}(t_{hj}^{(\text{old})},\sigma_t^2)$.
    \item[(h)] Update $m_{hjk}^{(l)}$ by Gibbs, where $m_{hjk}^{(l)}\sim 2\times\text{Ber}\left(\frac{1}{1+\xi_{hjk}^{(l)}}\right)$-1.
    \item[(i)] Update $\xi_{hjk}$ by random walk metropolis with a normal proposal $\mathcal{N}(\xi_{hjk}^{(\text{old})},\sigma_\xi^2)$.
\end{enumerate}

The random walk step sizes $\sigma_\eta,\, \sigma_\xi$ are fixed at 0.1, $\sigma_\mu$ is fixed at 0.5 and $\sigma_t$ is set dynamically (based on acceptances of updated thresholds before burn-in) to $0.1\times 2^{z}$, where $z$ is an integer in $[-4,\,4]$. Values of $a,\,b$ in the gamma prior $\mathrm{Gamma}(a,\,b)$ for thresholds are set at 10, 10 respectively. Both gamma and inverse-gamma distributions, throughout this paper, are in shape-rate parameterization. For qDAGx with known or misspecified ordering, we use $2\times 10^4$ iterations with a burn-in of $10^4$ samples and save every $10^\mathrm{th}$ sample after the burn-in. For qDAGx with unknown ordering and also in real data application, we use 5000 MCMC samples with a burn-in of 2500 and save every $10^\text{th}$ sample after the burn-in. In all the three cases, the minimum number of MCMC samples are chosen such that they give desired convergence results in the least amount of time. Representative MCMC diagnostic plots for the three procedures we compare, $\mathrm{qDAGx}_0,\,\mathrm{qDAGx}_m$ and qDAGx, at the node $\Yvec_1$ and quantile level $\tau=0.5$, is presented in Fig.~\ref{trace_plot_label}. Similar plots were observed across all nodes at different quantile levels and problem dimensions.
\begin{figure}[!t]
    \centering
    \includegraphics[width = \textwidth]{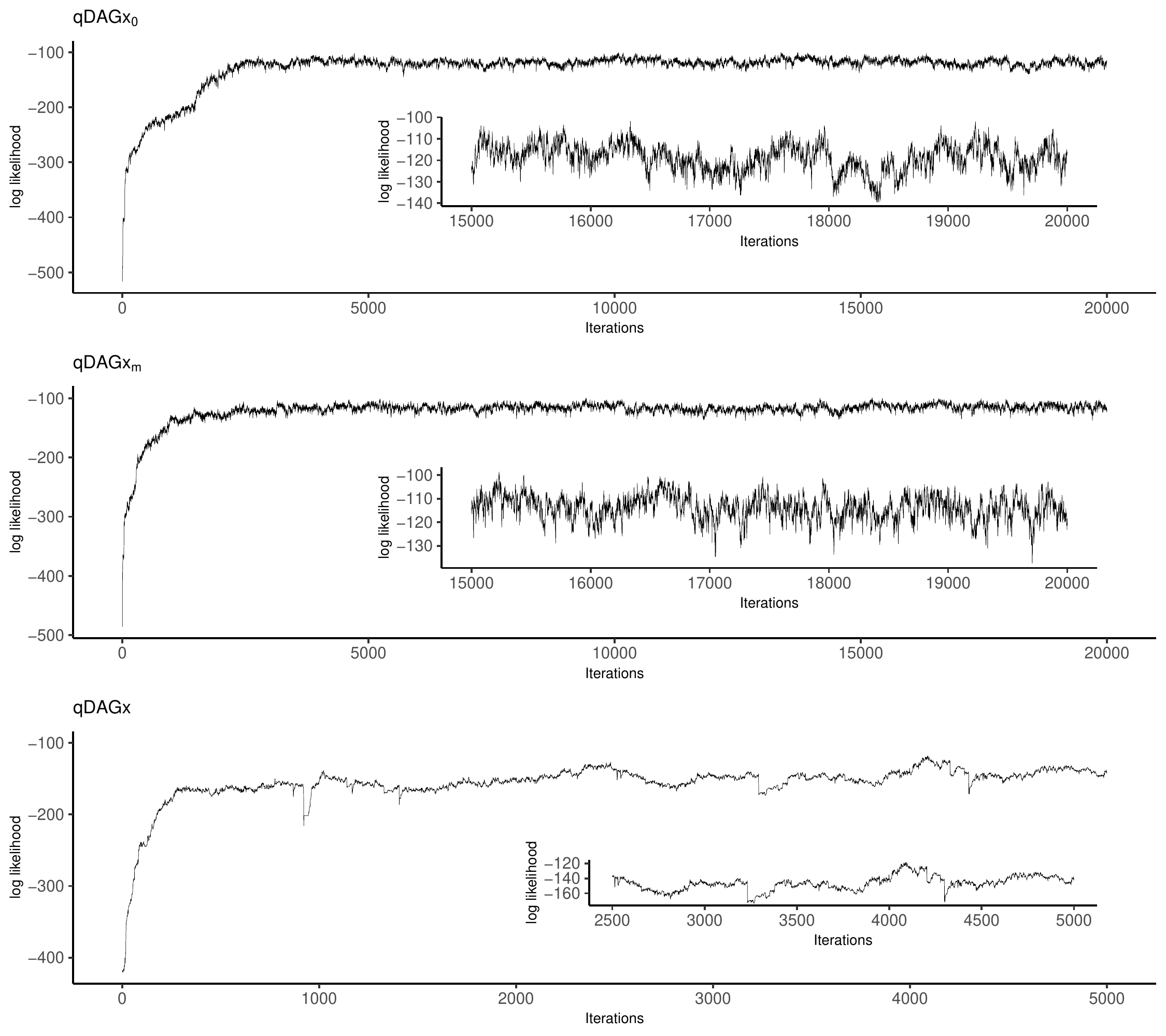}
    \caption{\label{trace_plot_label}
    Trace plot of log likelihood at the node $\Yvec_1$ vs. Iterations, for the three procedures compared in the simulations (Section~\ref{simulation_results}). Problem dimensions: $p=25,\,q=2,\, n=250$ and Kendall's' $\mathrm{T}$ for the misspecified sequence is 0.5.
    }
\end{figure}

\section{Variable Selection}
\label{Bayesian_estimation_and_variable_selection}
As outlined in Section~\ref{qDAGx_master_intro_section},  the edge $Y_{ih}\leftarrow Y_{ij}$ at quantile level $\tau$ exists, if the estimate of $\beta_{hj}^{(\tau)}(\Xvec_{i\cdot})$ is not zero. This information also translates into $(h,\, j)^\mathrm{th}$ entry in the adjacency matrix of $\mathcal{Q}\mathcal{G}_{i}^{(\tau)}$ being equal to $1$. 

First, when estimating the parameters in the case of `oracle', we have seen that the working likelihood can be factorized and the factored likelihood at a given node $\Yvec_h$ is obtained as in \eqref{factorised_likelihood_known_ordering}. Hence we can get the posterior samples of all parameters in the model, independently, for every $h\in\{1,\ldots,p\}$ and at every $\tau\in(0,\,1)$. This is not the case when the ordering of nodes in unknown as in~\eqref{DAG_working_likelihood}. In this case,  we start the sampling procedure with an empty graph $\mathcal{Q}\mathcal{G}^{(\tau)}_u$  i.e., with all entries in each $\mathcal{Q}\mathcal{G}^{(\tau)}_i$ set to zero, and then update the entries of the union-DAG based on the samples of $\beta_{hj}^{(\tau)}(\Xvec_{i\cdot})$ at every iteration. 

After sampling from the posterior at a given quantile level $\tau$, we detail the approach for variable selection used in this paper. First, depending on whether $\beta_{hj}^{(\tau)}(\cdot)$ is zero, we select the response variables. For the covariates, exact zeros do not arise under continuous shrinkage priors such as the horseshoe. Hence, we perform variable selection using the `pseudo-probabilities' of posterior inclusion \citep{carvalho2010horseshoe, datta2013asymptotic}, computed as: $1-1/(1+T_{hj}^2L_{hjk}^2) \in (0,1)$, where $T_{hj},\, L_{hjk}$ are the scale parameters in the prior of $\alphavec_{hjk}^*$, and the computed probability gives the `non-linear rate'. Similarly, computing the inclusion probability using scale parameters in the prior of $\alpha_{hjk}^0$, gives the `linear rate'. Taking a maximum of these two rates (linear and non-linear), gives the pseudo-probability of posterior inclusion of the effect of $\Xvec_k$, on the edge  $\Yvec_h\leftarrow\Yvec_j$.  

We use false discovery rate (FDR) control at 10\% in selecting the edges (response variables) and the covariates. First, we form Boolean matrices of appropriate dimensions whose entries are one at a posterior sample if the estimate of $\beta_{hj}^{(\tau)}(\Xvec_{i\cdot})\neq 0$, else zero; for a node $\Yvec_h$ at the given quantile level $\tau$. Comparing this Boolean matrix with the truth for true and false positives, at every posterior sample, gives two more Boolean matrices. And taking the average of these true and false positive Boolean matrices, over the dimension corresponding to the posterior samples, gives us the respective posterior probabilities of true and false positives, of the edges $\Yvec_h\leftarrow\Yvec_j,\,j\in\{1,\ldots,p\}\setminus\{h\}$ for a given node $\Yvec_h$ at the  quantile level $\tau$, in all the observations. Now consider a fine grid of thresholds $\in(0,\,1)$ and for every threshold in this grid, count the number of posterior probabilities (corresponding to true and false positives) greater than the threshold (which also corresponds to variable selection at this threshold); followed by the computation of FDR. Pick that threshold in the grid which has the FDR closest to 10\% and use it for the variable selection of the edges $\Yvec_h\leftarrow\Yvec_j,\,j\in\{1,\ldots,p\}\setminus\{h\}$, for a given node $\Yvec_h$ at the quantile level $\tau$. Repeating the same procedure with `pseudo probabilities' of posterior inclusion instead of the Boolean matrices, gives the variable selection results of the covariates, i.e., the influence of $\Xvec_k,\,k\in\{1,\ldots,q\}$ on the edges $\Yvec_h\leftarrow\Yvec_j,\,j\in\{1,\ldots,p\}\setminus\{h\}$, for a given node $\Yvec_h$ at the quantile level $\tau$.

\section{Proof of Theorem~\ref{identifiability_theorem}}
\label{proof_of_identifiability}

Let $\mathcal{Q}\mathcal{G}^{(\tau)}$ and $\mathcal{Q}\mathcal{G}^{'(\tau)}$ be two distinct quantile-DAGs parameterized by $\betavec^{(\tau)}$ and $\betavec^{(\tau)'}$, respectively. As we establish the identifiability for the population, we suppress the notation $i$, with respect to observations. The likelihood of quantile-DAG $\mathcal{Q}\mathcal{G}^{(\tau)}$ from Equation~\eqref{DAG_working_likelihood} can be written as:
\small
\begin{equation*}
\begin{split}
 \pi(\Yvec \mid \Xvec,\,\tau,\,\betavec^{(\tau)}) &= \prod_{h=1}^{p}\tau(1-\tau)\exp\left\{-\psi_\tau\left( Y_{h} -  \beta_{h0}^{(\tau)}(\Xvec)  - \sum_{j\neq h} Y_{j}\beta_{hj}^{(\tau)}(\Xvec)\right)\right\}\times \Ind(\mathcal{Q}\mathcal{G}^{(\tau)}\text{ is a DAG})\\
 & = \prod_{h=1}^{p}\tau(1-\tau)\exp\left\{-\psi_\tau\left( Y_{h} -  \beta_{h0}^{(\tau)}(\Xvec)  - \sum_{j\in pa(h)} Y_{j}\beta_{hj}^{(\tau)}(\Xvec)\right)\right\},
\end{split}
\end{equation*}
\normalsize

\noindent where $\psi_\tau(x) = \tau x\Ind(x\geq 0) +(\tau-1)x\Ind(x<0)$. Denote the likelihood of quantile-DAG parameterized by $\betavec^{(\tau)'}$ as $ \pi(\Yvec \mid \Xvec,\,\tau,\,\betavec^{(\tau)'})$. We prove that $ \pi(\Yvec \mid \Xvec,\,\tau,\,\betavec^{(\tau)})\neq  \pi(\Yvec \mid \Xvec,\,\tau,\,\betavec^{(\tau)'})$ by contradiction. Assume, 
\begin{equation}\label{eq:eql}
    \pi(\Yvec \mid \Xvec,\,\tau,\,\betavec^{(\tau)}) =  \pi(\Yvec \mid \Xvec,\,\tau,\,\betavec^{(\tau)'}),\,\forall\,\Ymat.
\end{equation}

Without loss of generality, assume that  $\mathcal{Q}\mathcal{G}^{(\tau)}$ is topologically sorted, i.e, whenever $Y_h\leftarrow Y_j$, we have $h<j$. Hence, the node $Y_1$ has no children, $ch(1)=\emptyset$. Also, for the sake of simplicity, we suppress $\Xmat,\,\tau$ while writing the coefficients $\beta_{hj}^{(\tau)}(\Xmat)$. Taking logarithms on both sides of \eqref{eq:eql} and differentiating it with respect to $Y_1$, we get,
\begin{equation}
\label{first_eq_iden}
    g_1(\tau) = g_1'(\tau) - \underset{k\in ch'(1)}{\sum} g_k'(\tau)\beta_{k1}',
\end{equation}
where,
\begin{equation*}
    \begin{matrix}
    g_h(\tau) = \begin{cases} \tau & \text{ if }Y_h \geq \beta_{h0}+\underset{j\in pa(h)}{\sum} Y_{j}\beta_{hj}\\
    \tau-1 & \text{ if }Y_h < \beta_{h0}+\underset{j\in pa(h)}{\sum} Y_{j}\beta_{hj}\end{cases},
 &   g_h'(\tau) = \begin{cases} \tau & \text{ if }Y_h \geq \beta_{h0}'+\underset{j\in pa'(h)}{\sum} Y_{j}\beta_{hj}'\\
    \tau-1 & \text{ if }Y_h < \beta_{h0}'+\underset{j\in pa'(h)}{\sum} Y_{j}\beta_{hj}'\end{cases}
        \end{matrix}.
\end{equation*}
Define two sets, 
\begin{equation*}
    A  =\{k\in ch'(1)\mid \beta_{k1}'\leq 0\}\text{ and } B  =\{k\in ch'(1)\mid \beta_{k1}'> 0\}.
\end{equation*}
Then if,
\small
\begin{equation*}
    Y_1\geq \text{max}\Bigg[\beta_{10}+\underset{j\in pa(1)}{\sum} Y_{j}\beta_{1j},\,\, \beta_{10}'+\underset{j\in pa'(1)}{\sum} Y_{j}\beta_{1j}',\,\,\underset{k\in ch'(1)}{\text{max}}\left\{\frac{1}{\beta_{k1}'}\left(Y_k - \beta_{k0}' - \underset{j\in pa'(k)\setminus\{1\}}{\sum}Y_k \beta_{kj}'\right)\right\}\Bigg],
\end{equation*} 
\normalsize
we have:
\begin{align}
\label{case_1}
    &g_k'(\tau) =\tau\text{, for }k\in A,\nonumber\\
    &g_k'(\tau) =\tau-1\text{, for }k\in B,\\
    &g_1(\tau) =g_1'(\tau)=\tau.\nonumber
\end{align}
Similarly, if,
\small
\begin{equation*}
    Y_1< \text{min}\Bigg[\beta_{10}+\underset{j\in pa(1)}{\sum} Y_{j}\beta_{1j},\,\, \beta_{10}'+\underset{j\in pa'(1)}{\sum} Y_{j}\beta_{1j}',\,\,\underset{k\in ch'(1)}{\text{min}}\left\{\frac{1}{\beta_{k1}'}\left(Y_k - \beta_{k0}' - \underset{j\in pa'(k)\setminus\{1\}}{\sum}Y_k \beta_{kj}'\right)\right\}\Bigg],
\end{equation*} 
\normalsize
we have:
\begin{align}
\label{case_2}
    &g_k'(\tau) =\tau-1\text{, for }k\in A,\nonumber\\
    &g_k'(\tau) =\tau\text{, for }k\in B,\\
    &g_1(\tau) =g_1'(\tau)=\tau-1.\nonumber
\end{align}
\normalsize
Therefore, \eqref{first_eq_iden} becomes, 
\begin{align*}
    0 & = \underset{k\in ch'(1)}{\sum} g_k'(\tau)\beta_{k1}'=\underset{k\in A}{\sum} g_k'(\tau)\beta_{k1}'+\underset{k\in B}{\sum} g_k'(\tau)\beta_{k1}'
\end{align*}    
which leads to
\begin{align*}
    0=\underset{k\in A}{\sum}\tau\beta_{k1}'+\underset{k\in B}{\sum} (\tau-1)\beta_{k1}'\text{(from~\eqref{case_1})} &\text{ and } 0= \underset{k\in A}{\sum}(\tau-1)\beta_{k1}'+\underset{k\in B}{\sum} \tau\beta_{k1}'\text{(from~\eqref{case_2})}.
\end{align*}
Hence, 
\begin{align*}
    \underset{k\in A}{\sum}\beta_{k1}' & =\underset{k\in B}{\sum} \beta_{k1}'.
\end{align*}

As $\beta_{k1}'\leq 0$ for $k\in A$ and $\beta_{k1}'> 0$ for $k\in A$, $\underset{k\in A}{\sum}\beta_{k1}' =\underset{k\in B}{\sum} \beta_{k1}'$ holds if and only if $\beta_{k1}'=0\,\forall\,k\in ch'(1)$. This further implies that $ch'(1)=\emptyset$. With this~\eqref{first_eq_iden} becomes, $g_1(\tau)=g_1'(\tau)$, which implies $pa(1)=pa'(1),\, \beta_{10}=\beta_{10}'\text{ and }\beta_{1j}=\beta'_{1j}\,\forall\,j\in pa(1)$. Finally, we can marginalize out $Y_1$ and we are left with quantile-DAGs $\mathcal{Q}\mathcal{G}^{(\tau)}\text{ and }\mathcal{Q}\mathcal{G}^{'(\tau)}$ with $p-1$ nodes and the likelihood $\pi(\Yvec_{-1} \mid \Xvec,\,\tau,\,\betavec^{(\tau)})$ becomes,
\small
\begin{equation*}
 \pi(\Yvec_{-1} \mid \Xvec,\,\tau,\,\betavec^{(\tau)}) = \prod_{h=2}^{p}\tau(1-\tau)\exp\left\{-\psi_\tau\left( Y_{h} -  \beta_{h0}^{(\tau)}(\Xvec)  - \sum_{j\in pa(h)} Y_{j}\beta_{hj}^{(\tau)}(\Xvec)\right)\right\}.
\end{equation*}
\normalsize
It is possible to marginalize out $Y_1$ while retaining a factorized form of the likelihood of the remaining variables $\Yvec_{-1}$ because, $ch(1)=ch'(1)=\emptyset$ and $\pi(\Yvec \mid \Xvec,\,\tau,\,\betavec^{(\tau)}) = \pi(Y_1 \mid \Yvec_{-1},\,\Xvec,\,\tau,\,\betavec^{(\tau)})\pi(\Yvec_{-1} \mid \Xvec,\,\tau,\,\betavec^{(\tau)})$. Similarly, the likelihood parameterized by $\betavec^{(\tau)'}$ can be factorized. Recursively starting with the node with no children e.g., node $Y_2$ in case of above likelihood, we get $pa(h)=pa'(h)\,\forall h$. Hence $\betavec^{(\tau)'}= \betavec^{(\tau)}\implies \mathcal{Q}\mathcal{G}^{(\tau)}\equiv\mathcal{Q}\mathcal{G}^{'(\tau)}$.
\section{Proof of Theorem~\ref{predective_consistency}}
\label{predictive_consistency_proof}
For a fixed $h$, we denote $\beta_{hj}^{(\tau)}$ as $\beta_j$, $\theta_{hj}^{(\tau)}$ as $\theta_{j}$, $\gamma_{hj}$ as $\gamma_j$, $\bgamma_h=\{\gamma_{hj}\}=\{\gamma_j\}$ as $\bgamma$ and $t_h$ as $t$.  From the marginal prior of $\beta_{hj}^{(\tau)}$ computed in Section~\ref{sec_predictive_consistency}, the prior on $(\beta_j\mid\gamma_j = 1)$ is,
\begin{equation*}
    \pi(\beta_j\mid\gamma_j = 1) = \mathcal{P}_\theta(\beta_j)\frac{\Pr(t<|\beta_j|\mid \beta_j)}{E_t\{\Pr(|\theta_j|> t\mid t)\}},
\end{equation*}
where $\mathcal{P}_\theta(\beta_j)$ is horseshoe density evaluated at $\beta_j$ and $t\sim\mathrm{Gamma}(a,b)$. As the horseshoe prior does not have a closed form density, we work with the proper density $(2\pi\sqrt{c})^{-1}\log(1+c/\beta_j^2)$ for $\mathcal{P}_\theta(\beta_j)$, referred to as the `horseshoe-like' prior by \citet{bhadra2021horseshoe}, and show the relevant result follows from the tight upper and lower bounds on the horseshoe density established by \citet{carvalho2010horseshoe}. We will establish the posterior consistency of fitted density as per the framework laid out by \citet{ghosal2000convergence}. The three required conditions from which the desired consistency result follows are:
\begin{itemize}
\setlength\itemsep{-0.1em}
    \item[(a)] The prior concentration rate of Kullback--Leibler (KL) $\epsilon_n^2$ neighborhood is at least $ \exp(-Cn\epsilon_n^2)$, for some $C>0$.
    \item[(b)] For a suitable sieve i.e., set of constraints on the parameters in the space of prior densities, the logarithm of the covering number (metric entropy) of the sieve is at most $ n\epsilon_n^2$.
    \item[(c)] The prior probability of the complement of the sieve is upper bounded by $ \exp(-c'n\epsilon_n^2)$, for some $c'>0$.
\end{itemize}
\vskip 0.25cm
\textbf{\underline{(a) Prior concentration rate of KL $\epsilon_n^2$ neighborhoods:}}
\vskip 0.25cm
Let $\pi_1 = \pi^*(\bY_h\mid \bY_{j\neq h};\,\betavec^*_h)\pi^*(\bY_{j\neq h})$ and $\pi_2 =\pi(\bY_h\mid \bY_{j\neq h};\,\bgamma,\,\betavec_{\bgamma_{h}})\pi^*(\bY_{j\neq h})$. For $K(\pi_1,\,\pi_2) = \int \pi_1\log(\pi_1/\pi_2)d\Ymat,\,V(\pi_1,\,\pi_2)= \int \pi_1\log^2(\pi_1/\pi_2)d\Ymat$, define $\mathcal{B}(\pi_1,\epsilon_n)= \{\pi_2:\,K(\pi_1,\,\pi_2)\leq \epsilon_n^2,\,V(\pi_1,\,\pi_2)\leq \epsilon_n^2\}$, the KL $\epsilon_n^2$ neighborhoods of the true density. To prove prior concentration rate, we need to prove, $\Pi_n(\mathcal{B}(\pi_1,\,\epsilon_n))\geq \exp(-Cn\epsilon_n^2)$, for some $C>0$; where $\Pi_n$ denotes the prior density on the coefficients and the subscript $n$ is to indicate that the parameters of the prior are functions of the sample size $n$. \citet[Proof of Theorem 3.1]{guha2020quantile} have proved that, for asymmetric Laplace likelihoods, the coefficients in the model $\betavec_{\bgamma_{h}}$, or simply $\betavec_{\bgamma}$, to be within $ \left(\beta_j^* \pm \frac{\eta\epsilon_n^2}{r_n}\right)$ for some $\eta>0$ (chosen appropriately by obeying Condition A6), when $\mathcal{B}(\pi_1,\,\epsilon_n)$ holds true. So, for the required prior concentration rate, we need to prove,
$$
\Pi_n\left\{\betavec_{\bgamma}=\{\beta_j\}:\,\beta_j\in\left(\beta_j^* \pm \frac{\eta\epsilon_n^2}{r_n}\right)\right\} = \underset{\beta_j \in \left(\beta_j^* \pm \frac{\eta\epsilon_n^2}{r_n}\right)}{\int \prod}\pi(\beta_j\mid\gamma_j = 1)d\beta_j > \exp(-Cn\epsilon_n^2),
$$ 
for some constant $C>0$. The product $\prod \pi(\beta_j\mid\gamma_j = 1)d\beta_j$,  is the product of the prior densities of the coefficients $\beta_j$ in the model $\bgamma$, and the notation remains same for the rest of the proof. Plugging in the chosen closed form density for $\mathcal{P}_\theta(\beta_j)$ and expanding the gamma cumulative distribution function $\Pr(t<|\beta_j|\mid\beta_j)$, we have:
\begin{equation}
\label{full_prior_expression}
    \pi(\beta_j\mid\gamma_j = 1) = \frac{1}{2\pi\sqrt{c}}\log\left(1+\frac{c}{\beta_j^2}\right)\exp(-b|\beta_j|)\sum_{k=a}^{\infty}\frac{(b|\beta_j|)^k}{k!}\frac{1}{E_t\{\Pr(|\theta_j|> t\mid t)\}}.
\end{equation}
For the denominator, we have:
\begin{align*}
    Pr(|\theta_j|> t\mid t) & = 2\int_{t}^\infty \frac{1}{2\pi\sqrt{c}}\log\left(1+\frac{c}{\theta_j^2}\right)d\theta_j\nonumber\\
    & \leq \frac{1}{\pi\sqrt{c}}\int_{t}^\infty \frac{c}{\theta_{j}^2}d\theta_j = \frac{\sqrt{c}}{\pi}\frac{1}{t}.\,\;(\text{using }\log(1+1/x^2)<1/x^2,\, x>0).
    \end{align*}
    This yields,
    \begin{align}
    \label{UB_of_horseshoe_required}
    E_t\{\Pr(|\theta_j|> t\mid t)\} & \leq \frac{\sqrt{c}}{\pi}\frac{b}{a-1}.
\end{align}
Using the bound on $E_t\{Pr(|\theta_j|> t\mid t)\}$ from \eqref{UB_of_horseshoe_required} and the prior density in Equation~\eqref{full_prior_expression}, we have the required integral
\begin{align}
    \underset{\beta_j \in \left(\beta_j^* \pm \frac{\eta\epsilon_n^2}{r_n}\right)}{\int \prod}\pi(\beta_j\mid\gamma_j = 1)d\beta_j & \geq \left(\frac{a-1}{2b}\right)^{r_n}\underset{\beta_j \in \left(\beta_j^* \pm \frac{\eta\epsilon_n^2}{r_n}\right)}{\int \prod}\frac{1}{c}\log\left(1+\frac{c}{\beta_j^2}\right)\exp(-b|\beta_j|)\sum_{k=a}^{\infty}\frac{(b|\beta_j|)^k}{k!}d\beta_j \label{step_1_integral}\\
    & \geq C_1^{r_n}\underset{\beta_j \in \left(\beta_j^* \pm \frac{\eta\epsilon_n^2}{r_n}\right)}{\int \prod}\frac{1}{c}\log\left(1+\frac{c}{\beta_j^2}\right)\exp\{-b|\beta_j|(1-K)\}|\beta_j|^a d\beta_j\nonumber\\
    &\quad\text{(using~\eqref{our_simplified_NLP}, for some }K\in(0,1)\text{ and absorbing constants }a,\,b\text{ into }C_1>0).\nonumber
\end{align}
It is easy to establish that the functions, $(1/c)\log(1+c/\beta_j^2)$ and $\exp\{-b|\beta_{j}|(1-K)\}$ are convex functions separately for $\beta_j>0$ and $\beta_j<0$. Hence, in the domain of integration, $\beta_j\in(\beta_j^*\pm \eta\epsilon_n^2/r_n)$, we can lower bound these functions, with their tangents at $\beta_j^*$, provided $\eta$ is chosen sufficiently small, such that $0\not\in(\beta_j^*\pm \eta\epsilon_n^2/r_n)$. Let $c_{j1},\,m_{j1}$ denote the intercept and slope of the tangent for the function $(1/c)\log(1+c/\beta_j^2)$ at $\beta_j^*$. Similarly, let  $c_{j2},\,m_{j2}$ denote the intercept and slope of the tangent for the function $\exp\{-b|\beta_{j}|(1-K)\}$ at $\beta_j^*$. We note that, $c_{j1},\,c_{j2}>0$ and $m_{j1}m_{j2}>0$ for any $\beta_j$. Therefore, we can lower bound the integral in the above equations as,
\begin{align*}
    \underset{\beta_j \in \left(\beta_j^* \pm \frac{\eta\epsilon_n^2}{r_n}\right)}{\int \prod}\pi(\beta_j\mid\gamma_j = 1)d\beta_j & \geq C_1^{r_n}\underset{\beta_j \in \left(\beta_j^* \pm \frac{\eta\epsilon_n^2}{r_n}\right)}{\int \prod}(c_{j1}+m_{j1}(\beta_j-\beta_j^*))(c_{j2}+m_{j2}(\beta_j-\beta_j^*))|\beta_j|^ad\beta_j\\
    & \geq C_1^{r_n}\underset{\beta_j \in \left(\beta_j^* \pm \frac{\eta\epsilon_n^2}{r_n}\right)}{\int \prod}(c_{j1}+m_{j1}(\beta_j-\beta_j^*))(c_{j2}+m_{j2}(\beta_j-\beta_j^*))|\beta_j-\beta_j^*|^a d\beta_j\\
    &\quad(\text{as }|\beta_j^*|>0\text{ and }\eta\epsilon_n^2/r_n\rightarrow 0)\\
    &=  C_1^{r_n}\underset{\beta_j \in \left(\beta_j^* \pm \frac{\eta\epsilon_n^2}{r_n}\right)}{\int \prod} (c_{j1}c_{j2} + m_{j1}m_{j2}(\beta_j-\beta_j^*)^2)|\beta_j-\beta_j^*|^a d\beta_j\\
    & =C_1^{r_n}\underset{j}{\prod}\left\{\frac{2c_{j1}c_{j2}}{a+1}\left(\frac{\eta\epsilon_n^2}{r_n}\right)^{a+1} + \frac{2m_{j1}m_{j2}}{a+3}\left(\frac{\eta\epsilon_n^2}{r_n}\right)^{a+3} \right\}\\
    & \geq C_1^{r_n}\underset{j}{\prod}\left(\frac{2c_{j1}c_{j2}}{a+1}+ \frac{2m_{j1}m_{j2}}{a+3} \right)\left(\frac{\eta\epsilon_n^2}{r_n}\right)^{a+3}\quad(\text{as }\eta\epsilon_n^2/r_n\rightarrow 0).
\end{align*}
From condition A4, we have $\underset{j}{\sum}|\beta_j^*|<\infty$ and $|\beta_j^*|>0$, which implies $\beta_j^*\,\forall\, j$ is bounded. Hence the constants $c_{j1}$ and $c_{j2}$ can be lower bounded by their respective functional evaluations, at $\text{max}(|\beta_j^*|)$. Similarly, $m_{j1}m_{j2}$ can be lower bounded by the product of slopes of the respective functions, at $\text{max}(|\beta_j^*|)$. Hence, using the lower bounds of $c_{j1},\,c_{j2}$ and $m_{j1}m_{j2}$, which are finite positive constants, and absorbing them into the constant $C_2>0$, we get:
\begin{align*}
    \underset{\beta_j \in \left(\beta_j^* \pm \frac{\eta\epsilon_n^2}{r_n}\right)}{\int \prod}\pi(\beta_j\mid\gamma_j = 1)d\beta_j & \geq C_2^{r_n}\left(\frac{\eta\epsilon_n^2}{r_n}\right)^{r_n(a+3)}\\
    & =\exp\left\{-C_3r_n-C_4 r_n\log\left(\frac{r_n}{\eta\epsilon_n^2}\right)\right\},\,C_4>0\text{ and }\text{for some constant }C_3.
\end{align*}
We know that $r_n<n\epsilon_n^2$ and so is $r_n\log r_n\leq r_n\log p_n\leq \bar{r}_n\log p_n \prec n\epsilon_n^2$ (condition A1), and $r_n\log(1/\epsilon_n^2)\leq \bar{r}_n\log(1/\epsilon_n^2)\prec n\epsilon_n^2$ (Condition A2). Hence, 
\begin{equation*}
     \underset{\beta_j \in \left(\beta_j^* \pm \frac{\eta\epsilon_n^2}{r_n}\right)}{\int \prod}\pi(\beta_j\mid\gamma_j = 1)d\beta_j \geq \exp(-Cn\epsilon_n^2),\,\text{ for some }C>0.
\end{equation*}
\vskip 0.25cm
\noindent\textbf{\underline{(b) Metric entropy of the chosen sieve:}}
\vskip 0.25cm
First, we state the definition of our chosen sieve as follows:
\begin{equation*}
    \underset{j}{\text{max }}|\beta_j|<D,\,D\sim n\epsilon_n^2\text{ and model size is upper bounded by }\bar{r}_n.
\end{equation*}
The proof for the bound on metric entropy of the sieve follows from~\citet[Checking condition (a)]{jiang2007bayesian}, by plugging in the size of coefficients ($C_n$ in \citet{jiang2007bayesian}) as $D$; radius of the $\ell^\infty$ balls covering the coefficients ($\delta$) as $\eta\epsilon_n^2/\bar{r}_n$; number of covariates $(K_n)$ as $p_n$ and from the Conditions A1, A2.
\vskip 0.25cm
\noindent\textbf{\underline{(c) Prior probability of the complement of the sieve:}}
\vskip 0.25cm
For this step, we need to get bounds for two probabilities which correspond to the prior probabilities of the complement of the conditions in the sieve. First is to prove $\int_{\underset{j}{\text{max }}|\beta_j|>D} \pi(\beta_j\mid\gamma_j = 1)<\exp(-c''n\epsilon_n^2)$ for some $c''>0$ and the next is to prove that the probability of model size is larger than $\bar{r}_n$ and smaller than $\exp(-c'''n\epsilon_n^2)$ for some $c'''>0$. We will prove the first bound and the second follows from~\citet[Proof of Theorem 1]{jiang2007bayesian}. To compute the first bound, we need a lower bound of the expectation $E_t\{Pr(|\theta_j|> t\mid t)\}$. So, starting with $Pr(|\theta_j|> t\mid t)$, it can be obtained as follows:
\begin{align*}
    Pr(|\theta_j|> t\mid t) & = 2\int_{t}^\infty \frac{1}{2\pi\sqrt{c}}\log\left(1+\frac{c}{\theta_j^2}\right)d\theta_j\\
    & \geq \frac{\sqrt{c}}{\pi}\int_{t}^\infty \frac{1}{\theta_{j}^2 + c}d\theta_j= \frac{\sqrt{c}}{\pi}\frac{1}{\sqrt{c}}\left(\frac{\pi}{2} - \mathrm{arctan}\frac{t}{\sqrt{c}}\right)\,\\
    & (\text{using }\log(1+x)>x/(1+x),\,x>0).
\end{align*}
Using the fact that for $x>0$, $\frac{\pi}{2} - \mathrm{arctan} x > \frac{1}{x} - \frac{1}{3x^3}$, followed by the required expectation over $t\sim\mathrm{Gamma}(a,\,b)$, we get:
\begin{equation*}
     E_t\{Pr(|\theta_j|> t\mid t)\} > \frac{\sqrt{c}}{\pi}\left(\frac{b}{a-1} - \frac{cb^3}{3(a-1)(a-2)(a-3)}\right)=\frac{b\sqrt{c}}{\pi(a-1)}\left(1-\frac{cb^2}{3(a-2)(a-3)}\right).
\end{equation*}
From condition A7 we have $a\in\mathrm{Z}^+,\,a>3$ and $c^2b<3/2$. So, we have the bound from the above Equation as, 
\begin{equation}
    \label{LB_of_horseshoe_required}
    \frac{1}{ E_t\{Pr(|\theta_j|> t\mid t)\}}<\xi\pi \frac{(a-1)}{b\sqrt{c}},\, \text{ for some }\xi>1.
\end{equation}
Now, rewriting the prior density from~\eqref{full_prior_expression}  by taking in the bounded expectation into proportionality, we have:
\begin{equation*}
\begin{split}
     \pi(\beta_j\mid\gamma_j = 1) & \propto \frac{1}{\sqrt{c}}\log\left(1+\frac{c}{\beta_j^2}\right)\exp(-b|\beta_j|)\sum_{k=a}^{\infty}\frac{(b|\beta_j|)^k}{k!}\\
     & \propto \frac{1}{\sqrt{c}}\log\left(1+\frac{c}{\beta_j^2}\right)\exp(-b|\beta_j|)\left(\exp(b|\beta_j|)-\sum_{k=0}^{a-1}\frac{(b|\beta_j|)^k}{k!}\right).
\end{split}
\end{equation*}
From the remainder theorem of Taylor's series, we can write, 
\begin{equation*}
    \exp(b|\beta_j|) = \sum_{k=0}^{a-1}\frac{(b|\beta_j|)^k}{k!} + \exp(Kb|\beta_j|)\frac{(b|\beta_j|)^a}{a!},\text{ where } K\in(0, 1).
\end{equation*}
Using the above result, we get:
\begin{equation}
\label{our_simplified_NLP}
    \pi(\beta_j\mid\gamma_j = 1) \propto \frac{1}{\sqrt{c}}\log\left(1+\frac{c}{\beta_j^2}\right)\exp\{-b|\beta_j|(1-K)\}|\beta_j|^a.
\end{equation}
With the simplified prior density in Equation~\eqref{our_simplified_NLP} and the lower bound for the expectation obtained in Equation~\eqref{LB_of_horseshoe_required}, we have: 
\begin{align}
    \int_{|\beta_j|>D} \pi(\beta_j\mid\gamma_j = 1)d\beta_j &< c_1 \xi\pi \frac{(a-1)}{b\sqrt{c}}\int_{|\beta_j|>D}\frac{1}{\sqrt{c}} \log\left(1+\frac{c}{\beta_j^2}\right)\exp\{-b|\beta_j|(1-K)\}|\beta_j|^a d\beta_j\label{step_3_integral}\\
    & < c_1 \xi\pi \frac{(a-1)}{b\sqrt{c}}\sqrt{c}\int_{|\beta_j|>D} \exp\{-b|\beta_j|(1-K)\}|\beta_j|^{a-2} d\beta_j\nonumber\\
    &\quad (\text{using }\log(1+x)<x,\, x>0).\nonumber
\end{align}
Using that fact that $|\beta_j|^{a-2}<\exp(c_2|\beta_j|)$, for sufficiently large $D$ and for an appropriately chosen constant $c_2$ such that, $0<c_2<b(1-K)$, we get 
\begin{align}
\label{complement_prob_sieve}
    \int_{|\beta_j|>D} \pi(\beta_j\mid\gamma_j = 1)d\beta_j &< \xi\pi \frac{(a-1)}{b} \int_D^\infty\exp\{-|\beta_j|(b(1-K) - c_2)\}d|\beta_j|\nonumber\\
    & = \xi\pi \frac{(a-1)}{b} \Bigg[\frac{\exp\{-|\beta_j|(b(1-K) - c_2)\}}{-(b(1-K) - c_2)}\Bigg]_{D}^{\infty}\nonumber\\
    & =  \xi\pi \frac{(a-1)}{b(b(1-K)-c_2)} \exp\{-D(b(1-K) - c_2)\}\nonumber\\
    & <\exp(-c_3 D) \sim \exp(-c_4 n\epsilon_n^2)\,(\because\, D \sim n\epsilon_n^2).
\end{align}
Using the elementary rules of order statistics, we can write:
\begin{align*}
    \int_{\underset{j}{\text{max }}|\beta_j|>D} \pi(\beta_j\mid\gamma_j = 1)d\beta_j& = 1- \left(1-\int_{|\beta_j|>D} \pi(\beta_j\mid\gamma_j = 1)d\beta_j\right)^{r_n}\\
    &\leq 2r_n\int_{|\beta_j|>D} \pi(\beta_j\mid\gamma_j = 1)d\beta_j\\
    &\leq  \exp(-c'' n\epsilon_n^2), \text{ for some }c''>0\,(\text{using A1 and A3}).
\end{align*}
This completes the step 3 and hence the proof of consistency of the fitted density under the chosen density $(2\pi\sqrt{c})^{-1}\log(1+c/\beta_j^2)$ for $\mathcal{P}_\theta(\beta_j)$. \citet{carvalho2010horseshoe} established the following tight bounds for the horseshoe prior:
\begin{equation}
\label{hs_bounds}
    \frac{1}{\sqrt{c}(2\pi)^{3/2}}\log\left(1+\frac{4c}{\theta_j^2}\right) <\pi_{HS}(\theta_j\mid c)< \frac{2}{\sqrt{c}(2\pi)^{3/2}}\log\left(1+\frac{2c}{\theta_j^2}\right). 
\end{equation}
Using the above bounds, we get the required proof with horseshoe prior on $\theta_j$ as follows:
\begin{enumerate}
    \item[(a)] For prior concentration rate of KL $\epsilon_n^2$ neighborhoods:
    \begin{align*}
        \pi(\beta_j\mid\gamma_j = 1) & = \mathcal{P}_\theta(\beta_j)\frac{\Pr(t<|\beta_j|\mid \beta_j)}{E_t\{\Pr(|\theta_j|> t\mid t)\}}\\
        & > \frac{1}{\sqrt{c}(2\pi)^{3/2}}\log\left(1+\frac{4c}{\theta_j^2}\right)\frac{\Pr(t<|\beta_j|\mid \beta_j)}{E_t\{\Pr(|\theta_j|> t\mid t)\}}.\\
        E_t\{\Pr(|\theta_j|> t\mid t)\} & = E_t\left(2\int_t^\infty\pi_{HS}(\theta_j\mid c)d\theta_j\right)\\
        & < E_t\left\{2\int_t^\infty \frac{2}{\sqrt{c}(2\pi)^{3/2}}\log\left(1+\frac{2c}{\theta_j^2}\right)d\theta_j\right\}\\
        & < E_t\left(\frac{2\sqrt{2c}}{\pi\sqrt{\pi}}\int_t^\infty \frac{1}{\theta_j^2}d\theta_j\right) = \frac{2\sqrt{2c}}{\pi\sqrt{\pi}}\frac{b}{a-1}.\\
        \implies  \underset{\beta_j \in \left(\beta_j^* \pm \frac{\eta\epsilon_n^2}{r_n}\right)}{\int \prod}\pi(\beta_j\mid\gamma_j = 1)d\beta_j & \geq \mathrm{C}_1^{r_n}\underset{\beta_j \in \left(\beta_j^* \pm \frac{\eta\epsilon_n^2}{r_n}\right)}{\int \prod}\frac{1}{4c}\log\left(1+\frac{4c}{\beta_j^2}\right)\exp(-b|\beta_j|)\sum_{k=a}^{\infty}\frac{(b|\beta_j|)^k}{k!}d\beta_j\\
        & \quad (\text{for some constant }\mathrm{C}_1>0).
    \end{align*}
    The integral above is analogous to the integral in~\eqref{step_1_integral}. So, the next steps in establishing the prior concentration rate of KL $\epsilon_n^2$ neighborhoods follow.
    \item[(b)] As metric entropy is the property of the chosen sieve, with the current assumptions on $D$ and $\bar{r}_n$, the metric entropy under the horseshoe prior also follows. 
    \item[(c)] For the prior probability of the complement of the sieve: 
    \begin{align*}
        \pi(\beta_j\mid\gamma_j = 1) & = \mathcal{P}_\theta(\beta_j)\frac{\Pr(t<|\beta_j|\mid \beta_j)}{E_t\{\Pr(|\theta_j|> t\mid t)\}}\\
        & < \frac{2}{\sqrt{c}(2\pi)^{3/2}}\log\left(1+\frac{2c}{\theta_j^2}\right)\frac{\Pr(t<|\beta_j|\mid \beta_j)}{E_t\{\Pr(|\theta_j|> t\mid t)\}}.\\
        E_t\{\Pr(|\theta_j|> t\mid t)\} & = E_t\left(2\int_t^\infty\pi_{HS}(\theta_j\mid c)d\theta_j\right)\\
        & > E_t\left\{2\int_t^\infty  \frac{1}{\sqrt{c}(2\pi)^{3/2}}\log\left(1+\frac{4c}{\theta_j^2}\right)d\theta_j\right\} \\
        & > E_t\left(\frac{1}{\pi\sqrt{2c\pi}}\int_t^\infty \frac{4c}{\theta_j^2 + 4c}d\theta_j\right) = E_t\left\{\frac{2}{\pi\sqrt{2\pi}}\left(\frac{\pi}{2} - \text{arctan}\frac{t}{2\sqrt{c}}\right)\right\}\\
        & > E_t\left\{\frac{2}{\pi\sqrt{2\pi}}\left(\frac{2\sqrt{c}}{t} - \frac{8c\sqrt{c}}{3t^3}\right)\right\} = \frac{2b\sqrt{2c}}{\pi\sqrt{\pi}(a-1)}\left(1-\frac{4cb^2}{3(a-2)(a-3)}\right).\\
 \implies \frac{1}{E_t\{\Pr(|\theta_j|> t\mid t)\}}  & < \xi'\frac{a-1}{b\sqrt{c}},\,\text{for some }\xi'>1\,(\text{from condition }A7).\\
        \implies   \int_{|\beta_j|>D} \pi(\beta_j\mid\gamma_j = 1)d\beta_j &< c_1' \xi' \frac{(a-1)}{b\sqrt{c}}\int_{|\beta_j|>D}\frac{1}{\sqrt{c}} \log\left(1+\frac{2c}{\beta_j^2}\right)\exp\{-b|\beta_j|(1-K)\}|\beta_j|^a d\beta_j\\
        &\quad(\text{using}~\eqref{our_simplified_NLP}).
    \end{align*}
    The integral above is analogous to the integral in~\eqref{step_3_integral}. So, the next steps in establishing the prior probability of the complement of the sieve follow.
\end{enumerate}

\section{Performance Metrics and Additional Numerical Results}
\label{extra_simulation_results}
An enumerated list of performance metrics which we use to compare results in simulations (Section~\ref{simulation_results}, Fig.~\ref{all_3_models_comparison}) is as follows:
\begin{enumerate}
    \item[(a)] $\mathrm{TPR}_{\mathrm{Y}}^\tau,\,\mathrm{TPR}_{\mathrm{X}}^\tau$: Average true positive rate of variable selection of response variables and covariates respectively, across all nodes $\Yvec_h,\,h\in\{1,\ldots,p\}$, at the given quantile level $\tau\in\{0.1,\ldots,0.9\}$. Analogously defined are $\mathrm{FPR}_{\mathrm{Y}}^\tau,\,\mathrm{FPR}_{\mathrm{X}}^\tau$ and $\mathrm{AUC}_{\mathrm{Y}}^\tau,\,\mathrm{AUC}_{\mathrm{X}}^\tau$.
    
    \item[(b)] $\Delta_F\betavec^\tau,\,\Delta_F\thetavec^\tau$: Scaled estimation norms~\eqref{estimation_norms}. The matrices $\betavec_{h}^{\tau,\,\text{est}},\, \thetavec_{h}^{\tau,\,\text{est}}$ are the collection of estimates of all QCIFs $\beta_{hj}^{(\tau)}(\cdot),\,\theta_{hj}^{(\tau)}(\cdot)$ , whose functional forms are evaluated at the posterior means of their respective parameters. In the case of unknown ordering, these estimation norms are scaled by a factor of $1/\sqrt{2}$ to adjust for twice the number of QCIFs estimated, when compared with the cases of known and misspecified ordering. 
    \begin{equation}
    \label{estimation_norms}
    \Delta_F\betavec^\tau = \sqrt{\frac{1}{n}\sum_{h=1}^{p}\Big\|\betavec_{h}^{\tau, \text{ est}} - \betavec_{h}^{\tau, \text{ true}}\Big\|^2_F} \text{ and }  \Delta_F\thetavec^\tau = \sqrt{\frac{1}{n}\sum_{h=1}^{p}\Big\|\thetavec_{h}^{\tau, \text{ est}} - \thetavec_{h}^{\tau, \text{ true}}\Big\|^2_F}.
\end{equation}

    \item[(c)] $\text{MSE}^\tau$: Adjusted mean squared error in quantile estimation~\eqref{MSE_known_ordering}. $\Qmat_{\Yvec_{h},\tau}^{\text{est}}$ is a $n\times 1$ dimensional vector, which contains the $\tau^\mathrm{th}$ quantile estimate of $\Yvec_h$, obtained by plugging in $\betavec_{h}^{\tau,\,\text{est}}$ in \eqref{patient_specific_model}, for all observations. The scale factor adjustment is done to normalize the effect of different number of parent nodes of $\Yvec_h,\, h\in\{1,\ldots,p\}$ (including the intercept term).
    \begin{equation}
    \label{MSE_known_ordering}
    \text{MSE}^\tau = \frac{1}{n}\sum_{h=1}^{p}\frac{\Big\|\Qmat_{\Yvec_{h},\tau}^{\text{true}} - \Qmat_{\Yvec_{h},\tau}^{\text{est}}\Big\|^2_2 }{\text{max}\Big\{1, \lfloor\frac{p-h}{5}\rfloor\Big\}+1}.
\end{equation}
\end{enumerate}

Additional numerical results for the experiments in Section~\ref{simulation_results} are presented in the Figures~\ref{supp_1_all_3_models_comparison} and~\ref{supp_2_all_3_models_comparison}. From both the figures we can see that qDAGx performs the best in estimation norms, which was also observed previously. For $p=25,\,50$ at $n=100$, we see that the performance of qDAGx in variable selection of response variables is close to that of the oracle and slightly worse, when it comes to covariates(Fig.~\ref{supp_1_all_3_models_comparison}). Whereas, when $p=100,\,q=2$ and $n=100$, qDAGx out-performs the oracle in variable selection of response variables. (Fig.~\ref{supp_2_all_3_models_comparison}\phantom{ }\ref{supp_2_all_3_models_comparison_a}). For covariates, variable selection results are relatively poorer at $p=50\text{ and }100$, when compared to the case of $p=25$ (Fig.~\ref{supp_1_all_3_models_comparison}~\ref{supp_1_all_3_models_comparison_b} and Fig.~\ref{supp_2_all_3_models_comparison}). And as observed in Section~\ref{simulation_results}, qDAGx with misspecified ordering performs the worst in most of the performance indicators. 

\begin{figure}[!bp]
\centering
    \begin{minipage}{\textwidth}
        \includegraphics[width=0.97\textwidth]{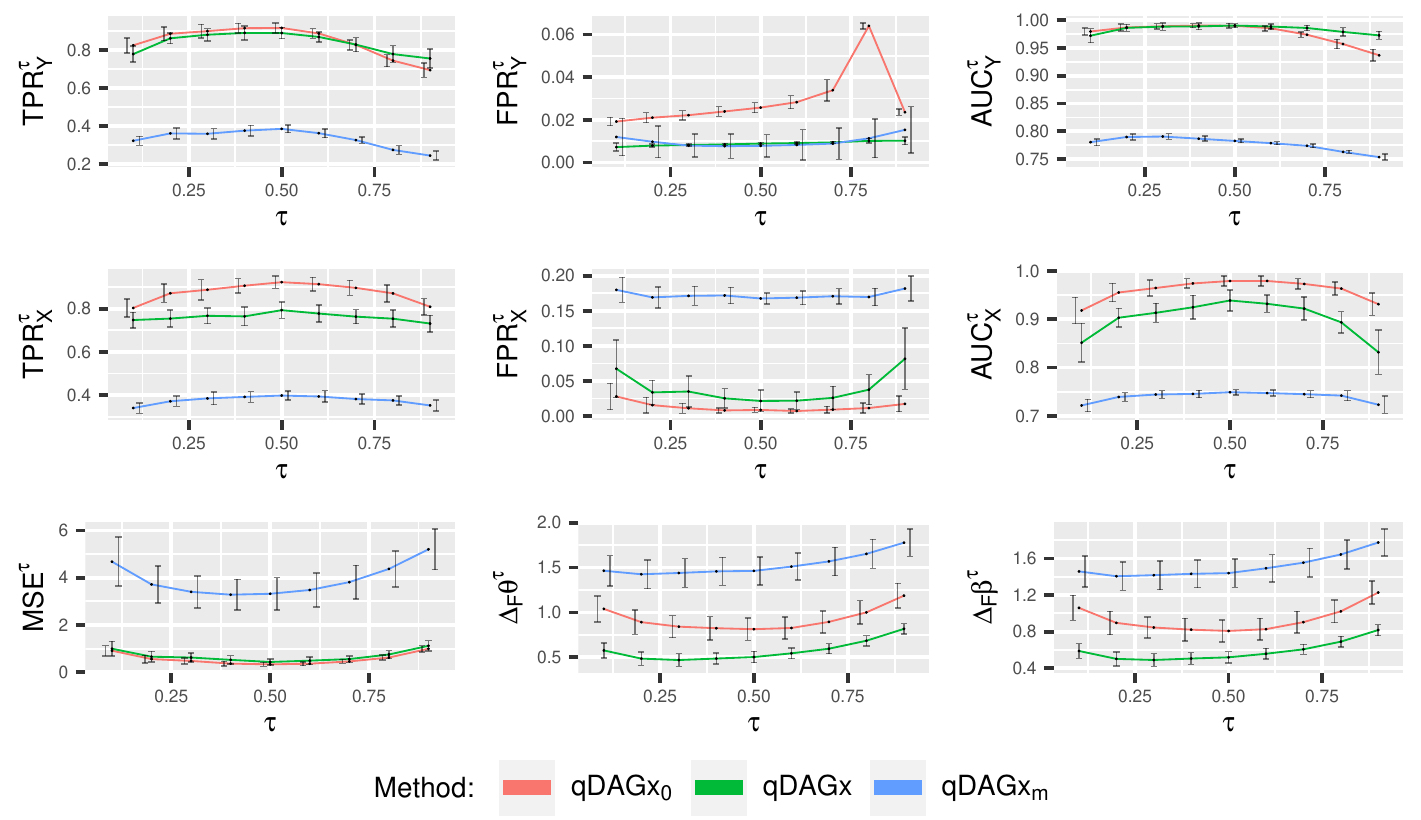}
        \Scaption{\label{supp_1_all_3_models_comparison_a} $p=25, q=2, n=100$. Kendall's' $\mathrm{T}$ for the misspecified sequence is 0.25}
    \end{minipage}\hfill
    \medskip
    \begin{minipage}{\textwidth}
        \includegraphics[width=0.97\textwidth]{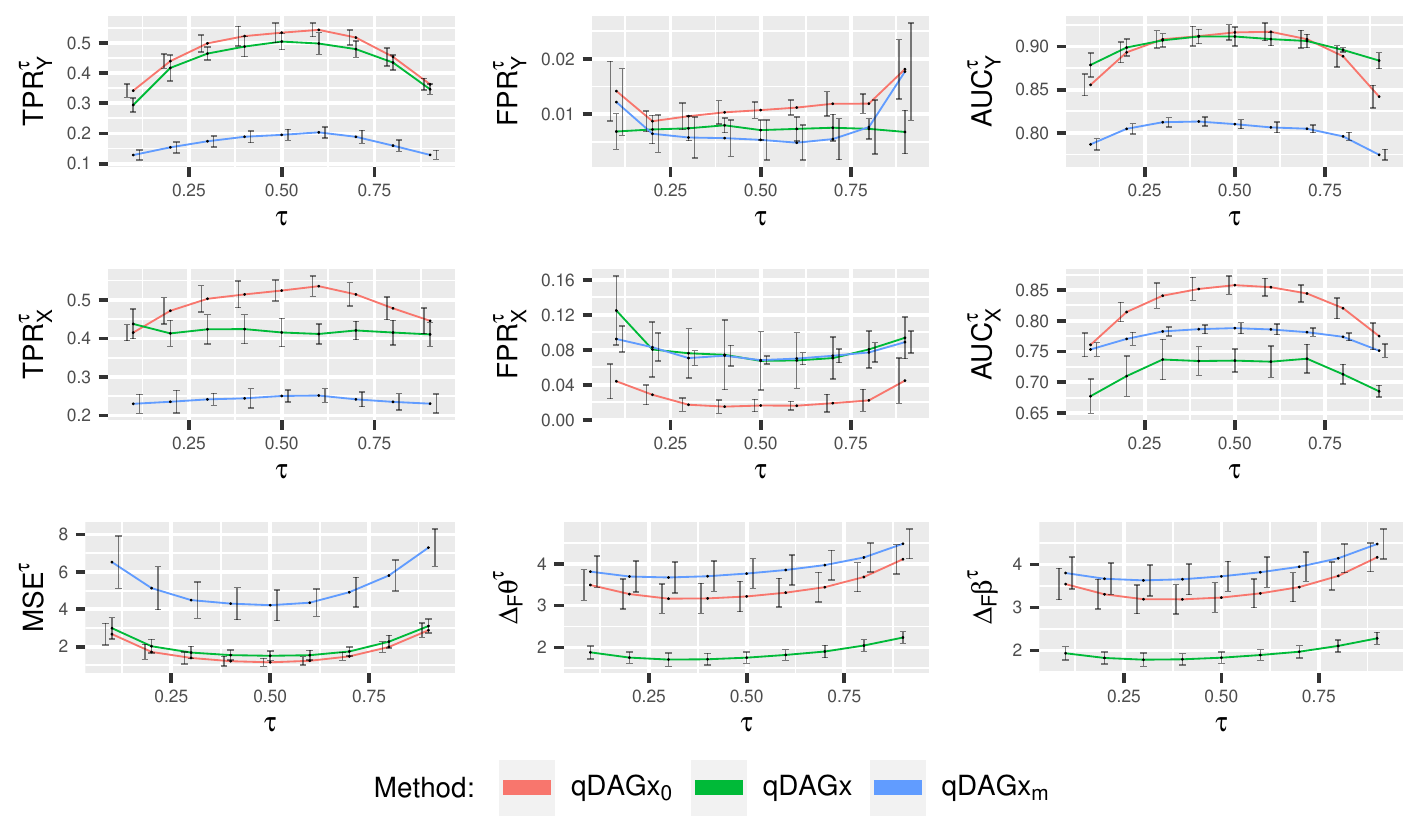}
        \Scaption{\label{supp_1_all_3_models_comparison_b} $p=50, q=5, n=100$. Kendall's' $\mathrm{T}$ for the misspecified sequence is 0.5}
    \end{minipage}\hfill
    \caption{\label{supp_1_all_3_models_comparison}Simulation results for two representative settings comparing the nine performance metrics between the quantile-DAG estimates of $\mathrm{qDAGx}_{0}$, qDAGx and $\mathrm{qDAGx}_\mathrm{m}$.}
\end{figure}

\begin{figure}[!bp]
\centering
    \begin{minipage}{\textwidth}
        \includegraphics[width=0.97\textwidth]{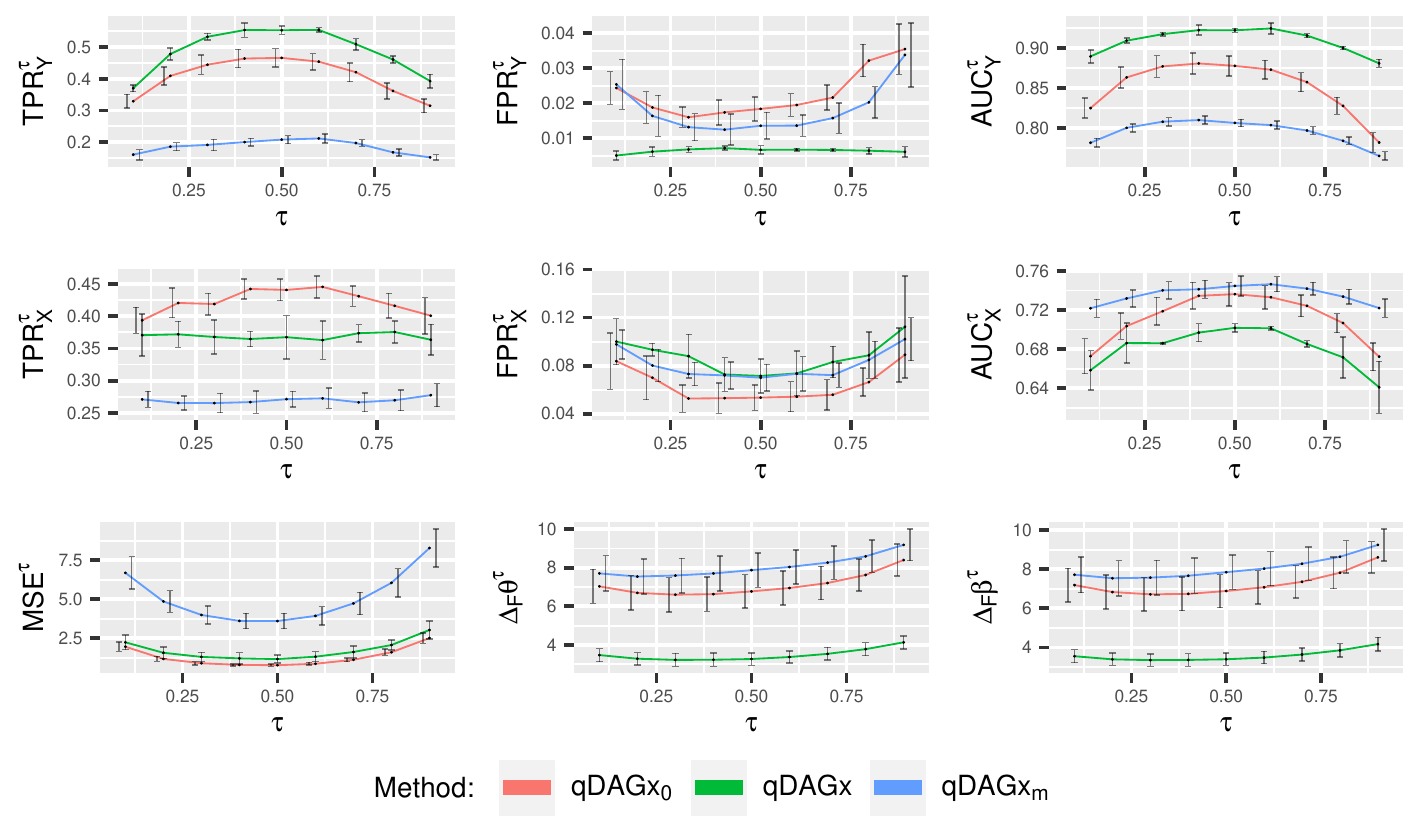}
        \Scaption{\label{supp_2_all_3_models_comparison_a} $p=100, q=2, n=100$. Kendall's' $\mathrm{T}$ for the misspecified sequence is 0.25}
    \end{minipage}\hfill
    \medskip
    \begin{minipage}{\textwidth}
        \includegraphics[width=0.97\textwidth]{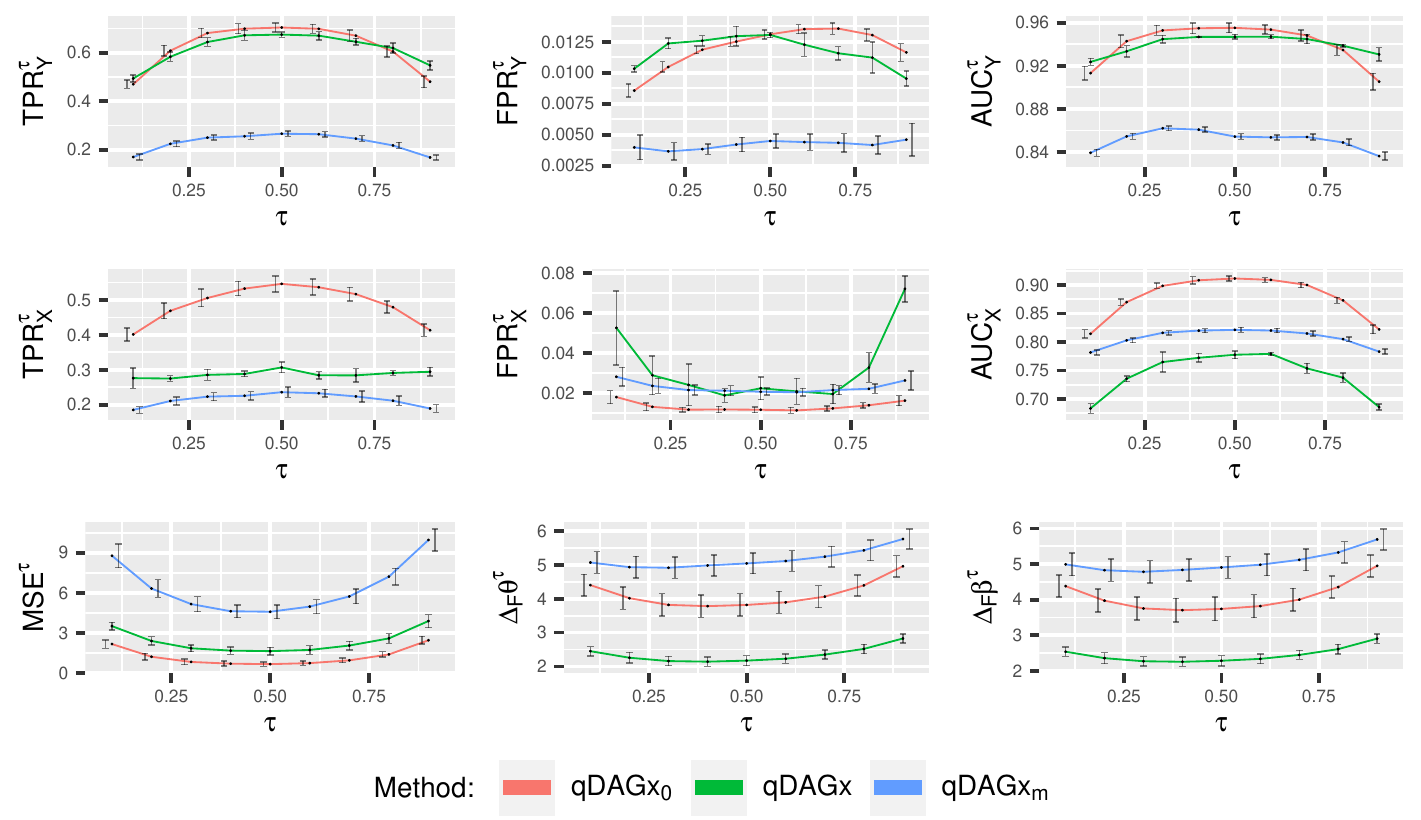}
        \Scaption{\label{supp_2_all_3_models_comparison_b} $p=100, q=5, n=250$. Kendall's' $\mathrm{T}$ for the misspecified sequence is 0.5}
    \end{minipage}\hfill
    \caption{\label{supp_2_all_3_models_comparison}Simulation results for two representative settings comparing the nine performance metrics between the quantile-DAG estimates of $\mathrm{qDAGx}_{0}$, qDAGx and $\mathrm{qDAGx}_\mathrm{m}$.}
\end{figure}

\section{Supplemental Comparisons with Lasso Penalized Quantile Regression}
\label{supplement_comparision_Lasso_QR}

The lasso penalized quantile regression or `lasso-QR'~\citep{wu2009variable} models the conditional quantiles as, 
\begin{equation}
\label{lasso_QR}
     Q_{Y_{ih}}(\tau\mid Y_{ij}) =  \beta_{h0}^{(\tau)} + \sum_{j=h+1}^{p} Y_{ij}\beta_{hj}^{(\tau)},
\end{equation}
where $\beta_{h0}^{(\tau)},\, \beta_{hj}^{(\tau)}$ are scalars without any functional form. In its simplest form, \eqref{lasso_QR} is not dependent on covariates. In order to accommodate them, we expand the model in \eqref{lasso_QR} as follows:
\begin{equation}
\label{lasso_QR_with_interaction}
     Q_{Y_{ih}}(\tau\mid T_{ij},\, \Xvec_{i\cdot}) =  \beta_{h0}^{(\tau)} + \sum_{j=h+1}^{p} Y_{ij}\beta_{hj}^{(\tau)} + \sum_{j=h+1}^{p}\sum_{k=1}^{q} Y_{ij}X_{ik}\beta_{hjk}^{(\tau)}.
\end{equation}
We incorporate interaction terms between response variables  and the covariates, in  \eqref{lasso_QR_with_interaction}. With lasso penalty on the coefficients, the likelihood for \eqref{lasso_QR_with_interaction} is
\begin{equation*}
\begin{split}
    \mathcal{L} & \propto  \sum_{i=1}^n-\psi_\tau\left( Y_{ih} -   \beta_{h0}^{(\tau)} - \sum_{j=h+1}^{p} Y_{ij}\beta_{hj}^{(\tau)} - \sum_{j=h+1}^{p}\sum_{k=1}^{q} Y_{ij}X_{ik}\beta_{hjk}^{(\tau)} \right) \\
   & \quad\quad -\lambda\left(|\beta_{h0}^{(\tau)}| + \sum_{j=h+1}^{p} |\beta_{hj}^{(\tau)}| +  \sum_{j=h+1}^{p}\sum_{k=1}^{q}|\beta_{hjk}^{(\tau)}| \right),
\end{split}
\end{equation*}
where $\lambda$ is the tuning parameter. It is important to note that $\beta_{hj}^{(\tau)}\text{ and }\beta_{hjk}^{(\tau)}$ are same for all observations. Hence what we get is a population level quantile-DAG that is not individualized. Regarding variable selection, we say that there is an edge between $\Yvec_h\text{ and }\Yvec_j$ if the estimated value of $\beta_{hj}^{(\tau)}\neq 0 \text{ or for at least one }k\text{, estimate of }\beta_{hjk}^{(\tau)}\neq 0$. Similarly we say that the $k^\text{th}$ covariate effects the edge between  $\Yvec_h\leftarrow \Yvec_j$ if the estimate of $\beta_{hjk}^{(\tau)}\neq 0$. With these two rules for variable selection, computing true and false positive rates is straightforward. Area under the ROC curves can be computed by tuning $\lambda$, getting true, false positive rates at different values of $\lambda$ and hence the AUC. Mean squared error of quantile estimation is computed as in \eqref{MSE_known_ordering} where the estimated quantiles come from plugging in the penalized estimates of $\beta_{h0}^{(\tau)},\, \beta_{hj}^{(\tau)}\text{ and }\beta_{hjk}^{(\tau)}$ in \eqref{lasso_QR_with_interaction}. It is important to note that $\Delta_F\betavec^\tau$ and $\Delta_F\thetavec^\tau$ cannot be computed for lasso-QR because estimating $\beta_{hj}^{(\tau)}(\Xvec_{i\cdot})$ is not in the scope of the model. 
\begin{figure}[!t]
    \centering
    \includegraphics[width=\textwidth]{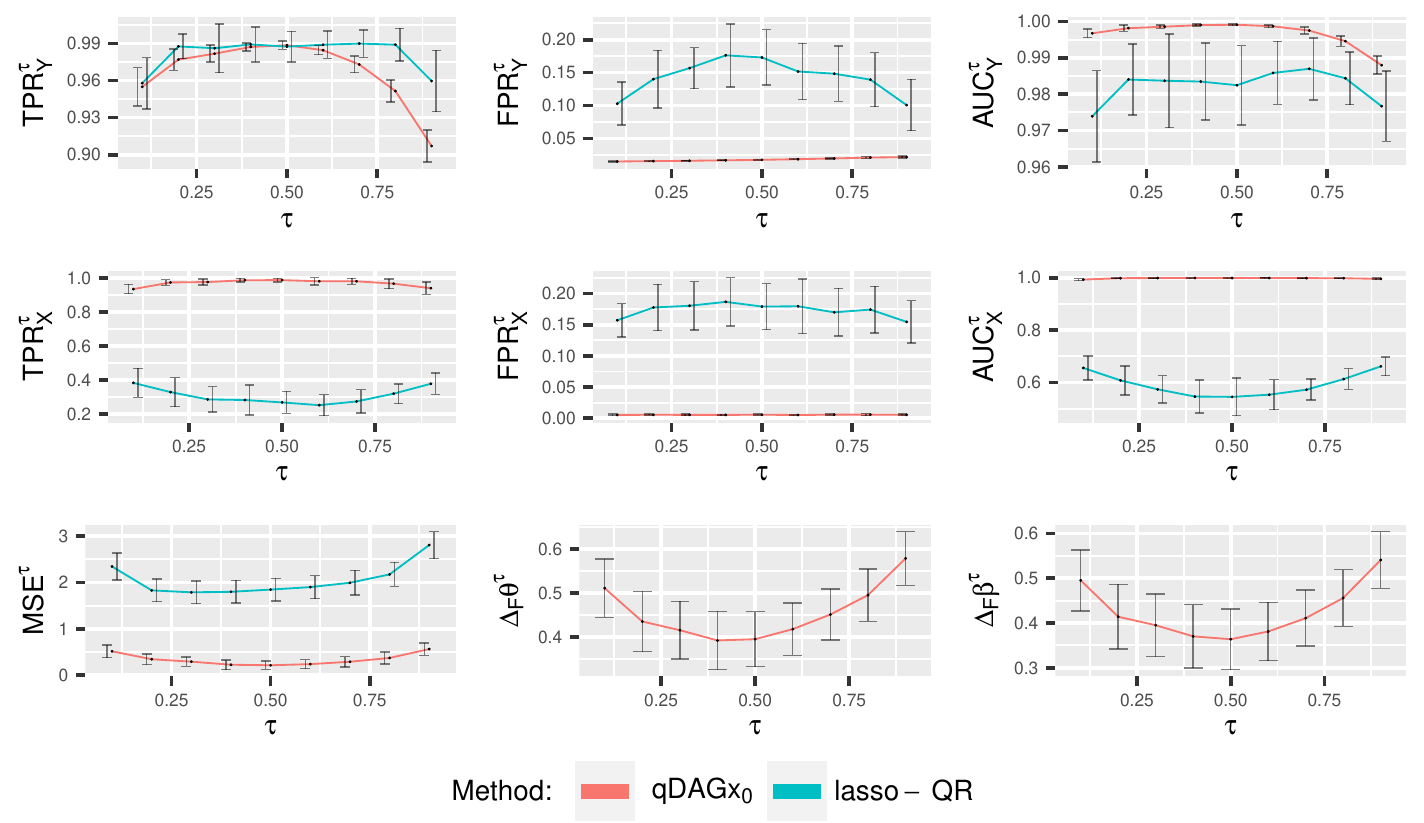}
        \caption{\label{lasso_qr_comparison}Comparison of performance measures of quantile DAG estimates by $\mathrm{qDAGx}_0$ and lasso-QR, when $p=25,\,q=5,\text{ and }n=250$.}
\end{figure}

With a representative simulation, we compare the results of $\mathrm{qDAGx}_0$ and lasso-QR in Fig. \ref{lasso_qr_comparison} when $p=25,\,q=5\text{ and }n=250$. A two level variable selection model (of response variables and covariates), with DAG discovery, using conditions analogous to ~\eqref{union_DAG_Condition}, under the lasso prior, is unexplored in the literature. Hence, we do not compare the results from the ad hoc model of lasso-QR in \eqref{lasso_QR_with_interaction} with that of qDAGx. In Fig.~\ref{lasso_qr_comparison}, one can see that though lasso-QR has slightly better true positive rates in variable selection of response variables, it performs poorly in terms of false positive rates and area under ROC curve. The performance in variable selection of covariates, and the mean squared error in quantile estimation, is extremely poor when compared with $\mathrm{qDAGx}_0$.
\section{Additional Results on the Lung Cancer Data}
\label{extra_real_data_analysis}
We present the names of the 67 proteins considered in Table~\ref{protein_number_mapping}, mapping between protein pathways and colors in Table~\ref{pathway_color_mapping} and representative visualizations of aggregated quantile-DAGs for LUSC in Fig.~\ref{LUSC_tau_1_5_9_results}.

    \begin{table}[!tbh]
    \caption{\label{protein_number_mapping}Map between node numbers and protein names. Node numbers are used instead of the protein names in Fig.~\ref{rand_patient_LUAD_tau_159}, \ref{LUAD_tau_1_5_9_results} and  \ref{LUSC_tau_1_5_9_results}.\\}
        \centering
        \resizebox{\textwidth}{!}{
        \begin{tabular}{|cc|cc|cc|cc|cc|cc|cc|}
         \hline
    			1  & BAK1    & 11 & MYH11 & 21 & PCNA    & 31 & TP53 & 41 & ATK1S1 & 51 & MAPK14 & 61 & MTOR   \\ 
    			2  & BAX     & 12 & \begin{tabular}[c]{@{}c@{}}RAB11A, \\ RAB11B\end{tabular} & 22 & FOXM1   & 32 & RAD50 & 42 & TSC2 & 52 & RPS6KA1  & 62 & RPS6   \\ 
    			3  & BID     & 13 & CTNNB1 & 23 & CDH1    & 33 & RAD51 & 43 & INPP4B & 53 & YBX1 & 63 & RB1    \\ 
    			4  & BCL2L11 & 14 & GADPH & 24 & CLDN7   & 34 & XRCC1 & 44 & PTEN & 54 & EGFR     & 64 & ESR1   \\ 
    			5  & CASP7   & 15 & RBM15 & 25 & TP53BP1 & 35 & FN1 & 45 & ARAF & 55 & ERBB2    & 65 & PGR    \\ 
    			6  & BAD     & 16 & CDK1 & 26 & ATM     & 36 & CDH2 & 46 & JUN & 56 & ERBB3    & 66 & AR     \\ 
    			7  & BCL2    & 17 & CCNB1 & 27 & CHEK1   & 37 & COL6A1 & 47 & RAF1 & 57 & SHC1     & 67 & GATA3  \\ 
    			8  & BCL2L1  & 18 & CCNE1 & 28 & CHEK2   & 38 & SERPINE1 & 48 & MAPK8 & 58 & SRC &    & \\ 
    			9  & BIRC2   & 19 & CCNE2 & 29 & XRCC5   & 39 & \begin{tabular}[c]{@{}c@{}}ATK1, ATK2, \\ ATK3\end{tabular} & 49 & \begin{tabular}[c]{@{}c@{}}MAPK1, \\ MAPK3\end{tabular} & 59 & EIF4EBP1 &    &        \\ 
    			10 & CAV1    & 20 & CDKN1B & 30 & MRE11A  & 40 & \begin{tabular}[c]{@{}c@{}}GKS3A, \\ GKS3B\end{tabular}     & 50 & MAP2K1 & 60 & RPS6KB1  &    &    \\
       \hline
        \end{tabular}
        }
    \end{table}

\begin{table}
\caption{\label{pathway_color_mapping}Map between pathways and colors. Colors are used instead of pathway names in  Fig.~\ref{rand_patient_LUAD_tau_159}, \ref{LUAD_tau_1_5_9_results} and  \ref{LUSC_tau_1_5_9_results}.\\}
    \centering
    \resizebox{\textwidth}{!}{
    \begin{tabular}{|cc | cc | cc | cc | cc | cc |}
    \hline
    Apoptosis & \cellcolor{pathway1} & Breast reactive  & \cellcolor{pathway2} & Cell cycle    & \cellcolor{pathway3} & Core reactive   & \cellcolor{pathway4} & \begin{tabular}[c]{@{}c@{}}DNA\\ damage response\end{tabular}    & \cellcolor{pathway5} & EMT & \cellcolor{pathway6} \\\hline
     PI3K/AKT & \cellcolor{pathway7} & RAS/MAPK  & \cellcolor{pathway8} & RTK   & \cellcolor{pathway9} &   TSC/mTOR  & \cellcolor{pathway10} &  Hormone receptor   & \cellcolor{pathway11} &  \begin{tabular}[c]{@{}c@{}}Hormone signaling\\ (Breast)\end{tabular}  & \cellcolor{pathway12} \\
     \hline
    \end{tabular}
    }
\end{table}

\begin{figure}[bp!]
    \centering
    \begin{minipage}{0.5\textwidth}
        \centering
        \includegraphics[scale=1]{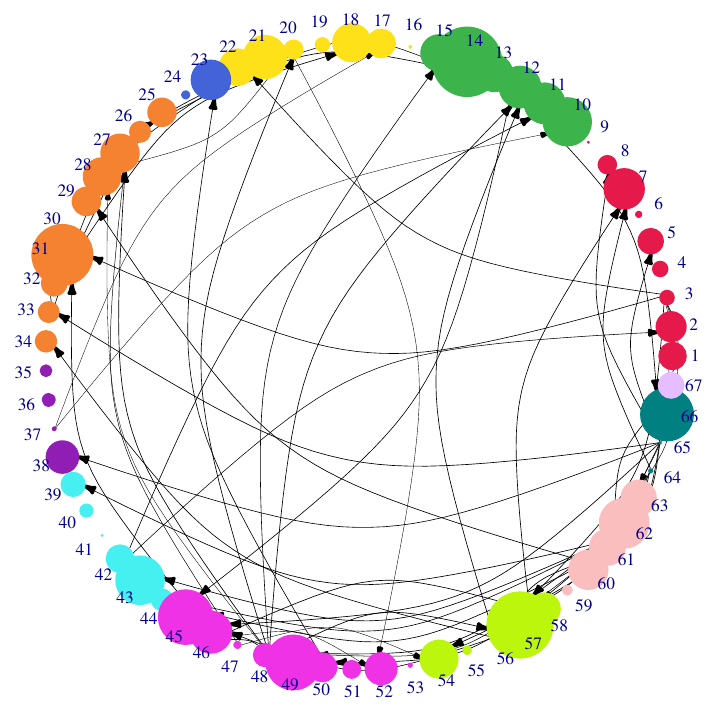}
        \Scaption{\label{LUSC_tau_1}}
    \end{minipage}\hfill
    \begin{minipage}{0.5\textwidth}
        \centering
         \includegraphics[scale =1]{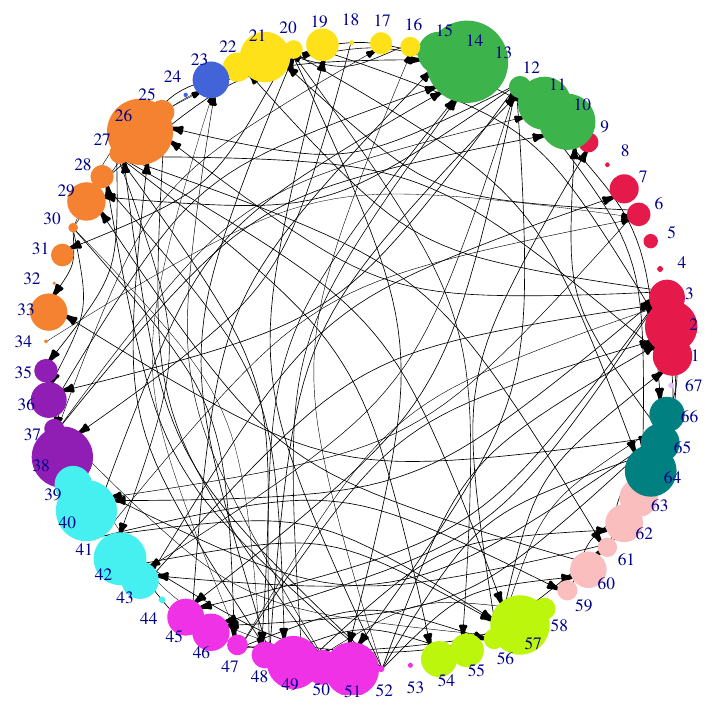}
        \Scaption{\label{LUSC_tau_5}}
    \end{minipage}
        \begin{minipage}{0.5\textwidth}
        \centering
        \includegraphics[scale=1]{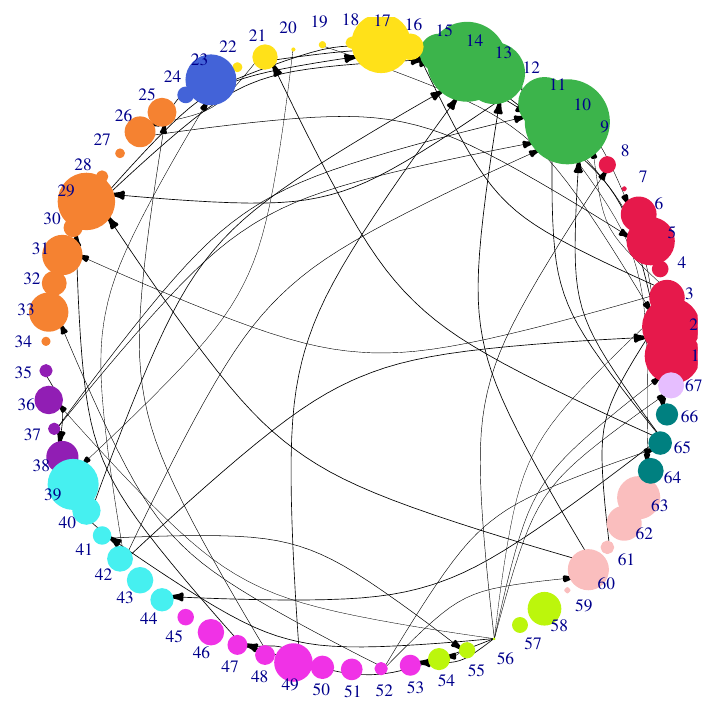}
        \Scaption{\label{LUSC_tau_9}}
    \end{minipage}\hfill
     \begin{minipage}{0.5\textwidth}
        \centering
        \includegraphics[scale=1]{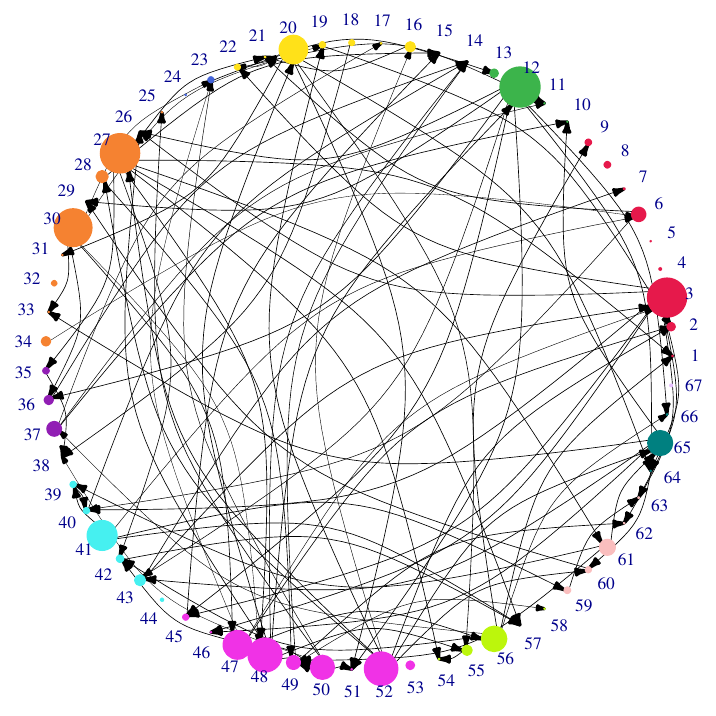}
        \Scaption{\label{LUSC_out_deg_tau_5}}
    \end{minipage}
    \caption{\label{LUSC_tau_1_5_9_results} Panels (a), (b), (c) show aggregated quantile-DAGs, $\Evec^{(\tau)}_{\text{LUSC}}$, for $\tau=0.1,\,0.5\text{ and }0.9$ respectively. Panel (d) shows $\Evec^{(0.5)}_{\text{LUSC}}\,$,  when the node size is proportional to out-degree of nodes. In all panels, nodes are colored according to the pathway to which they belong. The map between node colors and pathway names is given in Supplementary Table~\ref{pathway_color_mapping}. Note that there are some proteins which belong to multiple pathways; such proteins are just assigned to one of the pathways for the sake of clear visualization.}
\end{figure}

\end{document}